\documentclass{article}

\usepackage[english]{babel}
\usepackage{bm,hyperref,xcolor}
\usepackage[fleqn]{amsmath}
\usepackage{amsfonts}
\usepackage{amsgen,amsthm}
\usepackage{amssymb}
\usepackage{amsbsy}
\usepackage{bbm}
\usepackage{graphicx}
\usepackage{booktabs}
\usepackage{dcolumn} % per allineare i decimali nelle tabelle 
\usepackage[a4paper,margin=3cm]{geometry}
\usepackage{appendix}
\usepackage{lastpage}
\usepackage{enumerate}
\usepackage{textcomp}
\usepackage{calc}
\usepackage{ifthen}
\usepackage{natbib}
\usepackage{accents}
\newlength{\dhatheight}

\hypersetup{colorlinks = true, linkbordercolor = white, linkcolor = blue, citecolor = blue}
\providecommand\bnabla{\bm{\nabla}}
\providecommand\der{{d}}
\renewcommand{\vec}[1]{\ensuremath{\mathbf{#1}}} % for vectors
 
\providecommand\bnabla{\gvec{\nabla}}
\newtheorem{thm}{Theorem}
\newtheorem{lem}[thm]{Lemma}

\renewcommand{\bar}[1]{\overline{#1}}

\newcommand{\aderv}[2]{\frac{\partial {#1}}{\partial {#2}}}
\newcommand{\adervo}[2]{\frac{\der {#1}}{\der {#2}}}

\newcommand{\adervs}[2]{\frac{\partial^2 {#1}}{\partial {#2}^2}}
\newcommand{\adervso}[2]{\frac{\der^2 {#1}}{\der {#2}^2}}

\newcommand{\equa}[1]{Eq.~(\ref{#1})}

\newcommand{\equas}[1]{Eqs.~(\ref{#1})}

\newcommand{\equasa}[2]{Eqs.~(\ref{#1}){ }and{ }(\ref{#2})}

\newcommand{\ket}[1]{\big| #1 \big>} % for Dirac bras
\newcommand{\bra}[1]{\big< #1 \big|} % for Dirac kets
\newcommand{\braket}[2]{\big< #1 \vphantom{#2} \big|
  #2 \vphantom{#1} \big>} % for Dirac brackets
\newcommand{\eqn}[2]{\begin{gather}
#1
\label{#2}
\end{gather}
}

\newcommand{\spl}[2]{\begin{gather}
\displaybreak[2]
\begin{split}
#1
\end{split}
\label{#2}
\end{gather}
}

\newcommand{\gat}[2]{\begin{subequations}\label{#2}\begin{gather}
#1
\end{gather}\end{subequations}
}

\numberwithin{equation}{section}

\title{\bf An Introduction to Quantum Mechanics\\
{\large \ldots\ \emph{for those who dwell in the macroscopic world}}}
\author{\textbf{\em Antonio Barletta}  \\[-3pt]
	Department of Industrial Engineering \\[-4pt] 
Alma Mater Studiorum Universit\`a di Bologna  
	}

\date{---\ Lecture Notes\ ---\\{\small \today}}

\begin{document}

\maketitle

\vspace{0.5cm}

\section*{Foreword}
There is a huge number of excellent and comprehensive textbooks on quantum mechanics. They mainly differ for the approach, more or less oriented to the formalism rather than to the phenomenology, as well as for the topics covered. These lectures have been based to a large extent on the classical textbook by \citet{Gasiorowicz1974}. I must confess that the main reason for my choice of  \citet{Gasiorowicz1974} is affective, as it was the textbook where I first learned the basic principles of quantum mechanics. Beyond my personal taste, I now recognize that \citet{Gasiorowicz1974} is still a very good textbook on quantum mechanics, with a rigorous theoretical approach accompanied by a wide collection of applications. If the textbook by Gasiorowicz was my main basis, I have taken much also from other textbooks such as \citet{Phillips2003}, as well as from the excellent classical treatise by \citet{Dirac1981}. In order to avoid complications in the mathematics and in the notation, the topic is presented in these notes with reference to one--dimensional systems, with just a few marginal extensions to the three--dimensional formulation. 

\tableofcontents

\vspace{2cm}

\vfill\eject
\section{The origin of quantum physics}
A fundamental concept of classical physics is the particle, a point--like mass that moves along a trajectory in the three--dimensional space. The position and the instantaneous velocity of the particle can be determined with an arbitrarily high precision at every time, and their evolution is perfectly predictable from the equations of motions. The position and the velocity as functions of time can be known if one fixes the initial state of the particle, and if one knows the forces acting on the particle.

Another fundamental concept of classical physics is the wave. The electromagnetic field is a wave: an entity not localized in a specified position and with a velocity of propagation, also called phase velocity.

Waves and particles are independent paradigms of classical physics: a particle is not a wave and a wave is not a particle. As suggested by \citet{Phillips2003}, ``\emph{the laws of particle motion account for the material world around us and the laws of electromagnetic fields account for the light waves which illuminate this world}''.

This conception of natural world started to change at the beginning of twentieth century, as new experimental evidence was provided by the study of the black--body radiation, of the photoelectric effect, of the Compton effect, of the diffraction and interference properties of electrons, and of the spectra of atomic radiation.

\subsection{The black--body radiation} 
A result well--known from the studies initiated by Gustav Kirchhoff in 1859, is that the intensity of electromagnetic radiation emitted by a perfect absorber (\emph{black body}) is a universal function of the temperature $T$ and of the frequency of radiation $\nu$. In practice, the black--body radiation is the radiation from a small hole in an isothermal cavity. The approach of classical physics provided a theoretical formula to express the energy density (energy per unit volume) of the radiation inside the cavity in the frequency range $\nu$ to $\nu + d\nu$, 
\eqn{
u (\nu, T)\ d \nu = \frac{8\pi \nu^2}{c^3}\ k T\ d \nu ,
}{1}
where $k = 1.3807 \times 10^{-23} J/K$ is Boltzmann constant, and $c = 2.998 \times 10^8 m/s$ is the speed of light in vacuum. \equa{1} is now well--known as the Rayleigh--Jeans law. This law has the unphysical feature that, when integrated over all the possible frequencies, yields a diverging result. Moreover, it turns out to be in agreement with experimental data only for low frequencies. A completely different approach was needed to get a satisfactory agreement with experimental data, and this was done by Max Planck in 1900, who obtained
\eqn{
u (\nu, T)\ d \nu = \frac{8\pi h}{c^3}\ \frac{\nu^3}{e^{h\nu/(k T)} - 1}\ d \nu .
}{2}
\equa{2} contains a new constant of physics, \emph{Planck's constant},
\eqn{
h = 6.6261 \times 10^{-34} J s .
}{3}
One can easily check, by a Taylor expansion with small values of $h\nu/(k T)$, that \equa{2} coincides with \equa{1} at low frequencies. \equa{2} displays an excellent agreement with the experimental results, it can be integrated over all the frequencies with a converging result, but what is more interesting is that \equa{2} was obtained by invoking assumptions beyond the framework of classical physics. Planck assumed that radiation could be absorbed and emitted in ``quanta'' of energy given by $h \nu$. The electromagnetic radiation is treated as a gas of particles instead of a wave field. These particles, the quanta, are the photons. The same physical entity is treated as simultaneously having a wave and a particle character. Planck's hypothesis of photons is what we now call the wave--particle duality of electromagnetic radiation. 

\subsection{The photoelectric effect}
The photoelectric effect was discovered by Heinrich Hertz in 1887. This effect consists in the emission of electrons from a polished metal plate when irradiated by electromagnetic waves. The emission of electrons may take place (or not) depending on the frequency of the electromagnetic radiation. There exists a threshold frequency for the emission that depends on the metal. When the threshold is exceeded, the emission takes place and one has an electric current that increases with the intensity of the electromagnetic radiation.

What is incomprehensible in the framework of classical electromagnetism is that electrons may be emitted or not depending on the frequency of the radiation. The classical electromagnetism assigns an energy to the radiation that depends on the amplitude, but not on the frequency, of the waves. One can understand that the energy of the radiation incident on the metal plate may be transferred to the electrons trapped in the metal, so that the electrons may be extracted when a sufficient energy is received. However, on the basis of classical electromagnetism, one expects a threshold condition for the photoelectric effect to be formulated in terms of the radiation intensity and not of the radiation frequency. Let us call $W$ (the \emph{work function}) the energy required for an electron to escape the metallic bond, i.e. to separate the electron from the lattice of positively charged nuclei. Then, Planck's hypothesis of quanta allows one to express the kinetic energy of an escaping electron as
\eqn{
E_{kin} = h \nu - W .
}{4}
The threshold condition is thus $\nu > W/h$. The use of \equa{4} is the basis of Einstein's theoretical analysis of the photoelectric effect. Robert A. Millikan, skeptical about the theory of quanta, performed in 1915 the experiment that confirmed the validity of \equa{4}. He found that $W$ was of the order of several electron volts ($1\ eV = 1.602 \times 10^{-19} J$).

\subsection{Wave behaviour of particles and electron diffraction}
If the electromagnetic waves are also particles, the photons, is it conceivable to think of electrons as waves? This extrapolation of the wave--particle duality was the starting point of Louis de Broglie studies leading to the assumption that a particle having a momentum $p$ can be thought of as a wave having wave length
\eqn{
\lambda = \frac{h}{p} .
}{5}
This idea, first proposed in his 1924 doctorate thesis, was the object of experimental verification. In particular, in 1927, C. J. Davisson and L. H. Germer proved at Bell Telephone Laboratories in New York that diffraction of electrons from a crystal lattice could take place, thus yielding the first validation of de Broglie's idea. 

\subsection{The Bohr atom}
The discrete spectra of radiation emitted from atoms could not find a reasonable explanation by means of the atomic models based on the classical electromagnetism. The discrete spectra of radiation strongly suggest that the atomic energy levels are quantized. A transition between different levels yields an emitted radiation with a well specified wave length leading, experimentally, to a well specified spectral line. 

Niels Bohr in 1913 conjectured that the angular momentum of the electrons orbiting around the nucleus in an atom can be just an integer multiple of $\hbar = h/(2 \pi)$ (the reduced Planck constant). He further assumed that electrons do not emit radiation if not in a transition between different orbits. In such a transition, the frequency of the emitted or absorbed radiation is given by,
\eqn{
\nu = \frac{E - E'}{h} .
}{6}
The great success of Bohr's postulates in explaining the radiation of atoms, caused the development of this theory up to the so--called Sommerfeld--Wilson quantization rule,
\eqn{
\oint p\ dq = n h, \qquad n = 1, 2, 3,\ \ldots\ ,
}{7}
where the integral is over a closed path (orbit), and $p$ is the momentum canonically conjugated with the coordinate $q$ of a particle (see \ref{anamec}). \equa{7} was formulated in 1915. These theoretical studies have a role in the history of modern physics. However, they are to be considered as a mere phenomenological theory of atomic spectra, but not yet a reformulation of the basic concepts of mechanics.

\eject
\section{Wave packets and the uncertainty principle}
The basic ideas suggested by the failure of classical physics, when applied to microscopic objects (electrons, atoms, nuclei, \ldots) are condensed in what we called wave--particle duality. The challenge is to reconcile the concept of wave (a non--localized object) with that of particle (a localized object). Configurations of waves may exist where localization emerges in the form of wave packets. The basic mathematical theory is Fourier's analysis. \\
One may imagine a function $f(x)$ expressed as
\eqn{
f(x) = \int_{-\infty}^{\infty} dk\ g(k) e^{i k x} .
}{8}
\equa{8} defines $f(x)$ as a linear superposition of infinite standing plane waves with wavelength $\lambda = 2\pi/k$. Two neighbouring maxima of the real and imaginary parts of $e^{i k x}$ are separated by a distance $2\pi/k$. Each wave is ``weighted'' by the coefficient function $g(k)$. \\
One may consider a Gaussian weight function
\eqn{
g(k) = e^{-\alpha \left( k - k_0 \right)^2} ,
}{9}
where $\alpha > 0$. One can substitute \equa{9} into \equa{8} and evaluate the integral on the right hand side of \equa{8},
\eqn{
f(x) = \int_{-\infty}^{\infty} dk\ e^{-\alpha \left( k - k_0 \right)^2} e^{i k x} = e^{i k_0 x}  \int_{-\infty}^{\infty} dk'\ e^{-\alpha k'^2} e^{i k' x} = e^{i k_0 x} \sqrt{\frac{\pi}{\alpha}} \ e^{-x^2/(4 \alpha)} ,
}{10}
where the change of integration variable $k' = k - k_0$ has been employed.
On considering the square modulus of $g(k)$ and the square modulus of $f(x)$, 
\eqn{
|g(k)|^2 = e^{- 2 \alpha \left( k - k_0 \right)^2} , \qquad |f(x)|^2 =  \frac{\pi}{\alpha} \ e^{-x^2/(2 \alpha)} ,
}{11}
one realizes that we have a Gaussian signal both in $k$--space and in $x$--space. We may easily check that, when $k = k_0 \pm \Delta k/2$, where $\Delta k = 2/\sqrt{2 \alpha}$, the Gaussian signal in $k$--space drops to $1/e$ times its peak value. When $x = \pm \Delta x/2$, where $\Delta x = 2\sqrt{2\alpha}$, the Gaussian signal in $x$--space drops to $1/e$ times its peak value. If $\alpha$ becomes smaller and smaller, the signal in $k$--space increases its width $\Delta k$, while the signal in $x$--space decreases its width $\Delta x$. One may easily check that  
\eqn{
\Delta k \ \Delta x = 4 .
}{12}
The precise numerical value of the product is not important. What is important is that the product $\Delta k \ \Delta x$ is finite and independent of $\alpha$,
\eqn{
\Delta k \ \Delta x \sim O(1) .
}{13}
\equa{13} establishes an \emph{uncertainty principle} for the wave packets and its interpretation is straightforward: a highly localized wave packet in $k$--space, i.e. one with a small $\Delta k$, means a poorly localized wave packet in $x$--space, i.e. one with a large $\Delta x$, and vice versa. 
\begin{figure}[h!]
\begin{center}
\includegraphics[width=0.8\textwidth]{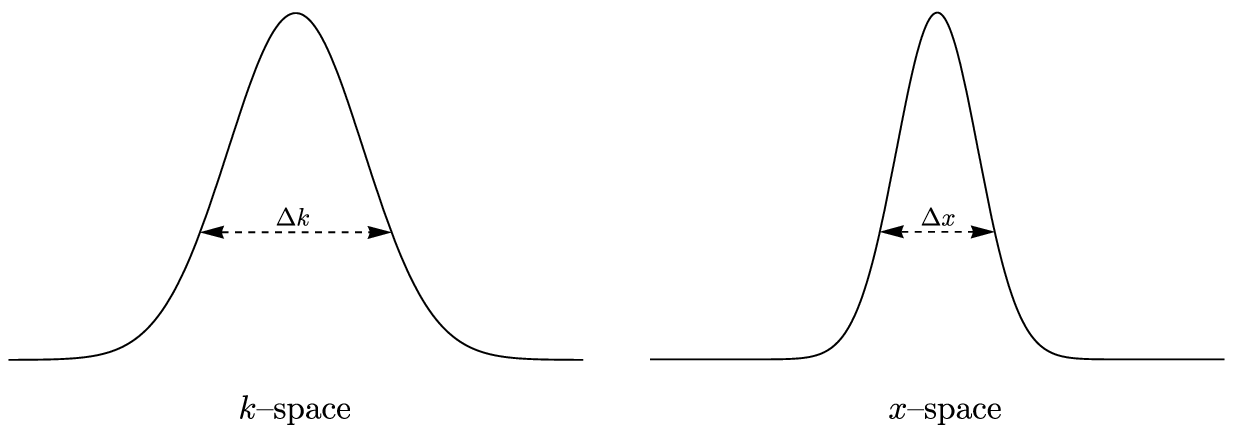}
\end{center}
\caption{Localization of wave packets in $k$--space and in $x$--space}
\label{fig1}
\end{figure}
It is not possible to reduce the width of the Gaussian signal both in $k$--space and in $x$--space.

Let us now consider a wave packet made up with the superposition of travelling plane waves. Then, \equa{8} is to be replaced with
\eqn{
F(x,t) = \int_{-\infty}^{\infty} dk\ g(k) e^{i \left( k x - \omega t \right)} ,
}{14}
where $\omega = 2 \pi \nu$ is the angular frequency, while $k$ is again related to the wavelength, $k=2\pi/\lambda$. The simplest case is that with $\omega = c k$, where $c$ is a constant. In this case, \equa{13} describes a superposition of plane waves having a constant phase velocity $c$. Then, a comparison between \equa{8} and \equa{14}, allows one to write
\eqn{
F(x,t) = f(x - c t) .
}{15}
The effect of the standing waves being replaced by travelling waves is just a rigid translational motion of the wave packet with a velocity $c$. No distortion of the wave packet is caused by the time evolution.

Our objective is the use of wave packets for the modelling of particles, i.e. of localized objects. We are therefore interested in assuming a general relationship $\omega = \omega(k)$. We may assume a wave packet strongly localized in $k$--space, with a marked peak at $k=k_0$, expressed by \equa{9}.
 
The strong localization in $k$--space suggests that one may express $\omega(k)$ as a Taylor expansion around $k=k_0$ truncated to lowest orders,
\eqn{
\omega(k) \approx \omega(k_0) + \left. \adervo{\omega}{k} \right|_{k=k_0}\!\! \left( k - k_0 \right) + \frac{1}{2} \left. \adervso{\omega}{k} \right|_{k=k_0}\!\! \left( k - k_0 \right)^2  .
}{16}
We use the notations
\eqn{
v_g = \left. \adervo{\omega}{k} \right|_{k=k_0} , \qquad \beta = \frac{1}{2} \left. \adervso{\omega}{k} \right|_{k=k_0} ,
}{17}
where $v_g$ is called the \emph{group velocity}.\\
We now substitute \equasa{9}{16} in \equa{14} and we obtain
\spl{
F(x,t) &= \int_{-\infty}^{\infty} dk\ e^{-\alpha \left( k - k_0 \right)^2} e^{i \left\{ k x - [\omega(k_0) + (k - k_0) v_g + (k - k_0)^2 \beta ] t \right\}} \\
&= \int_{-\infty}^{\infty} dk'\ e^{-\alpha k'^2} e^{i \left\{ (k_0 + k') x - [\omega(k_0) + k' v_g + k'^2 \beta ] t \right\}} \\
&= e^{i \left[ k_0  x - \omega(k_0) t \right]} 
\int_{-\infty}^{\infty} dk'\ e^{- \left( \alpha + i \beta t \right) k'^2} e^{i k' \left( x - v_g t \right)} 
\\
&= e^{i \left[ k_0  x - \omega(k_0) t \right]} 
\int_{-\infty}^{\infty} dk'\ e^{- \alpha' k'^2} e^{i k' x'} ,
}{18}
where $\alpha' = \alpha + i \beta t$ and $x' = x - v_g t$. We note that the integral appearing in \equa{18} is just the same as that evaluated in \equa{10}, with $\alpha$ replaced by $\alpha'$ and $x$ replaced by $x'$. Thus, we may write
\spl{
F(x,t) &= e^{i \left[ k_0  x - \omega(k_0) t \right]} 
\int_{-\infty}^{\infty} dk'\ e^{- \alpha' k'^2} e^{i k' x'} \\
&= e^{i \left[ k_0  x - \omega(k_0) t \right]} \sqrt{\frac{\pi}{\alpha + i \beta t}} \ e^{- \left( x - v_g t \right)^2/\left[4 \left( \alpha + i \beta t \right) \right]} .
}{19}
Again, we consider the square moduli of $g(k)$ and of $F(x,t)$ as in \equa{11}, 
\eqn{
|g(k)|^2 = e^{- 2 \alpha \left( k - k_0 \right)^2} , \qquad |F(x,t)|^2 =  \frac{\pi}{\sqrt{ \alpha^2 + \beta^2 t^2 }} \ e^{- \alpha \left( x - v_g t \right)^2/\left[2 \left( \alpha^2 + \beta^2 t^2 \right) \right]} .
}{20}
One recognizes Gaussian signals both in $k$--space and in $x$--space. The peak of the Gaussian signal in $x$ is located at $x = v_g t$, and thus it travels in the $x$--direction with the constant group velocity, $v_g$.

The width of the Gaussian signal in $k$--space is obviously $\Delta k = 2/\sqrt{2 \alpha}$, while the width of the Gaussian signal in $x$--space is now a function of time, 
\eqn{
\Delta x = 2 \sqrt{2 \alpha}\ \sqrt{1 + \frac{\beta^2 t^2}{\alpha^2}} .
}{21}
The width of the Gaussian signal in $x$--space increases in time. This means that the time evolution of the wave packet implies a spreading in $x$--space with a decreasing value at the peak position, $x = v_g t$. The latter feature can be easily inferred from \equa{20}.
\begin{figure}[h!]
\begin{center}
\includegraphics[width=0.8\textwidth]{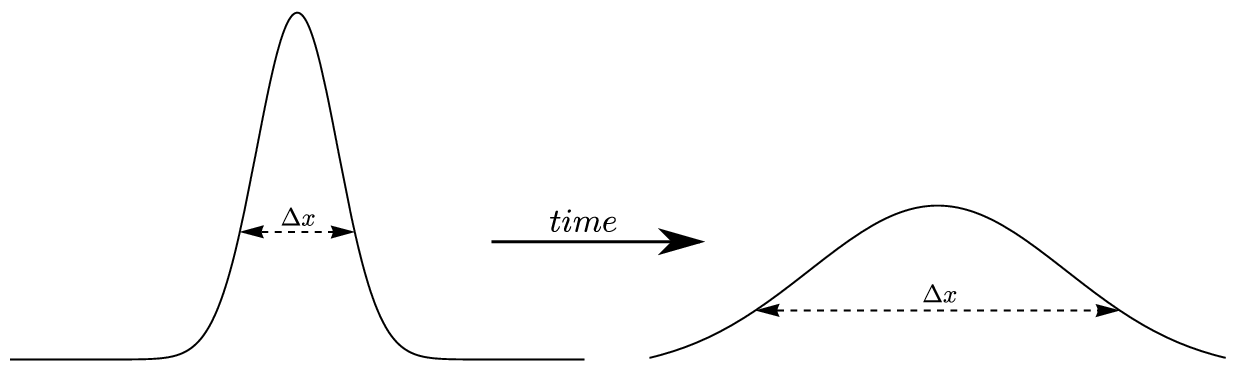}
\end{center}
\caption{Spreading of the wave packet in $x$--space}
\label{fig2}
\end{figure}\\
As a consequence, the uncertainty principle,
\eqn{
\Delta k \ \Delta x = 4\ \sqrt{1 + \frac{\beta^2 t^2}{\alpha^2}}\ > 4 ,\quad \Longrightarrow  \quad \Delta k \ \Delta x  \gtrsim O(1),
}{22}
now implies a more and more severe indetermination of the wave packets in $k$--space and in $x$--space as time progresses.

\subsection{Heisenberg uncertainty principle}
The goal of our analysis is the formulation of a wave theory of particles. Thus, we must intend the wave packet in $x$--space, travelling with a velocity $v_g$, as a particle with mass $m$, momentum $p$, and kinetic energy $E=p^2/(2m)$. This suggests the relationship
\eqn{
v_g = \adervo{\omega}{k} = \frac{p}{m} .
}{23}
Moreover, by invoking Planck's formula $E = h \nu = \hbar \omega$, we obtain
\eqn{
\omega = \frac{E}{\hbar} = \frac{p^2}{2 m \hbar} .
}{24}
\equasa{23}{24} yield,
\eqn{
\adervo{\omega}{p} = \frac{p}{m \hbar} = \frac{1}{\hbar}\ \adervo{\omega}{k} ,
}{25}
that can be satisfied on setting
\eqn{
k = \frac{p}{\hbar} .
}{26}
Since $\lambda = 2 \pi/k$, \equa{26} is consistent with de Broglie's assumption, \equa{5}. We now consider the expressions of $\omega$ and $k$ given by \equasa{24}{26} and the structure of the wave packet described by \equa{14} to express the \emph{wave function} of a particle as
\eqn{
\psi(x,t) = \frac{1}{\sqrt{2 \pi \hbar}} \int_{-\infty}^{\infty} dp\ \phi(p) e^{i \left( p x - E t \right)/\hbar} .
}{27}
The constant $1/\sqrt{2 \pi \hbar}$ plays just the role of an overall normalization of the wave packet, with no implications for the features described above. An important point is the interpretation of $\psi(x,t)$ as a model of a moving particle. For instance, the spreading of the wave packet associated with the time evolution cannot be interpreted as an increase in size of the particle, say an electron or a nucleus. On the other hand, the width of the wave packet in $x$--space describes the width of the region of space where the particle is likely to be found at a given instant of time. The spreading of the wave packet does not mean that the particle increases its size. It means that the region of space where the particle is likely to be found becomes larger and larger in the time evolution. Hence, its localization becomes more and more uncertain. Werner Heisenberg felt uneasy with the new role of the physical modelling of Nature, and he wrote \citep[see][page 94]{Baggott2011}:

\begin{quotation}
``When we know the present precisely, we can predict the future, is not
the conclusion but the assumption. Even in principle we cannot know
the present in all detail. For that reason everything observed is a selection
from a plenitude of possibilities and a limitation on what is possible in the
future. As the statistical character of quantum theory is so closely linked to 
the inexactness of all perceptions, one might be led to the presumption
that behind the perceived statistical world there still hides a ‘real’ world in
which causality holds. But such speculations seem to us, to say it explicitly,
fruitless and senseless. Physics ought to describe only the correlation of
observations. One can express the true state of affairs better in this way:
Because all experiments are subject to the laws of quantum mechanics~\ldots~, it
follows that quantum mechanics established the final failure of causality.''
\end{quotation} 

We are now ready for interpreting the uncertainty relationship, \equa{22}, in terms of the position $x$ and the momentum $p$ of the particle,
\eqn{
\Delta p \ \Delta x  \gtrsim O(\hbar),
}{28}
where \equa{26} has been used. The relationship between the widths of the Gaussian signals associated with $|\phi(p)|^2$ and $|\psi(x,t)|^2$ means that the uncertainty in the knowledge of the position of a particle is inversely proportional to the uncertainty in the knowledge of the momentum of the particle. The constant of proportionality is $O(\hbar)$.

The smallness of $\hbar$ ensures that \equa{28} is effective, and hence marks a breakup of classical physics, only for microscopic systems.  

\subsection{The wave equation for a free particle}
We start from \equa{27} and evaluate the time derivative of $\psi(x,t)$,
\spl{
\aderv{\psi(x,t)}{t} &= -\frac{i}{\hbar \sqrt{2 \pi \hbar}} \int_{-\infty}^{\infty} dp\ E \phi(p) e^{i \left( p x - E t \right)/\hbar}\\
&= -\frac{i}{2 m \hbar \sqrt{2 \pi \hbar}} \int_{-\infty}^{\infty} dp\ p^2 \phi(p) e^{i \left( p x - E t \right)/\hbar} ,
}{29}
where the relationship $E=p^2/(2m)$ for the kinetic energy has been used. We also evaluate the second derivative of $\psi(x,t)$ with respect to $x$,
\eqn{
\adervs{\psi(x,t)}{x} = -\frac{1}{\hbar^2 \sqrt{2 \pi \hbar}} \int_{-\infty}^{\infty} dp\ p^2 \phi(p) e^{i \left( p x - E t \right)/\hbar} .
}{30}
The integrals on the right hand sides of \equasa{29}{30} are coincident, so that one can write the equation
\eqn{
i \hbar\ \aderv{\psi(x,t)}{t} = - \frac{\hbar^2}{2 m}\ \adervs{\psi(x,t)}{x} .
}{31}
\equa{31} is a second--order partial differential equation that describes the motion of a free particle. In fact, no force and hence no potential energy is assumed to exist, as the energy is just the kinetic energy $E=p^2/(2m)$. Since \equa{31} was obtained from \equa{27}, the general solution of \equa{31} is
\eqn{
\psi(x,t) = \frac{1}{\sqrt{2 \pi \hbar}} \int_{-\infty}^{\infty} dp\ \phi(p) e^{i \left( p x - E t \right)/\hbar} .
}{32}
\equa{31} is a linear differential equation, so that any linear combination of two solutions of \equa{31}, $\psi_1(x,t)$ and $\psi_2(x,t)$, 
\eqn{
\psi(x,t) = c_1 \psi_1(x,t) + c_2 \psi_2(x,t) ,\qquad c_1 \in \mathbb{C},\ c_2 \in \mathbb{C},
}{33}
is still a solution of \equa{31}.

The wave function $\psi(x,t)$ can be determined at every instant of time $t >0$, by solving \equa{31}, if the wave function is known at the initial time $t=0$. In other words, the initial condition,
\eqn{
\psi(x,0) = \Psi_0(x) ,
}{34}
must be prescribed. This feature is a consequence of \equa{31} being first order in the time variable.

\subsection{The statistical interpretation of $\psi(x,t)$}
We pointed out that the spreading of the wave packet defined by
\eqn{
\psi(x,t) = \frac{1}{\sqrt{2 \pi \hbar}} \int_{-\infty}^{\infty} dp\ \phi(p) e^{i \left[ p x - \left(p^2/2 m\right) t \right]/\hbar} 
}{35}
does not mean an increase in size of the free particle, but just an increase in width of the region where the particle is likely to be found. This observation implies a statistical interpretation of the wave function $\psi(x,t)$. The wave function is complex--valued, so that the spreading of the wave packet was established on evaluating the width $\Delta x$ of the real distribution $|\psi(x,t)|^2$. Thus a reasonable guess is that
\eqn{
P(x,t)\ d x = |\psi(x,t)|^2\ dx 
}{36}
is the probability of finding the particle in the infinitesimal region between $x$ and $x + dx$. For this interpretation to be consistent, we must require that
\eqn{
\int_{-\infty}^{\infty} dx\ |\psi(x,t)|^2 = 1. 
}{37}
\equa{37} is a normalization condition of the wave function induced by its statistical interpretation: the probability of finding the particle somewhere is $1$. The normalization condition \equa{36} makes sense only if the integral on the left hand side of \equa{37} exists,
\eqn{
\int_{-\infty}^{\infty} dx\ |\psi(x,t)|^2 < \infty, 
}{38}
i.e. the wave function must be an element of the set $\mathcal{L}^2(\mathbb{R})$ of square integrable functions in the real domain.

The probability density $P(x,t)$, defined by \equa{36}, is conserved in the time evolution. This statement can be proved as follows. We evaluate the complex conjugate of the wave equation (\ref{31}),
\eqn{
i \hbar\ \aderv{\bar{\psi}(x,t)}{t} = \frac{\hbar^2}{2 m}\ \adervs{\bar{\psi}(x,t)}{x} ,
}{39}
where $\bar{\psi}(x,t)$ is the complex conjugate of $\psi(x,t)$.
We multiply \equa{31} by $\bar{\psi}(x,t)$, we multiply \equa{39} by $\psi(x,t)$ and sum the two resulting equations, so that we obtain
\spl{
&i \hbar \left( \bar{\psi}\ \aderv{\psi}{t} + \psi\ \aderv{\bar{\psi}}{t} \right) 
= -\frac{\hbar^2}{2 m} \left( \bar{\psi}\ \adervs{\psi}{x} - \psi\ \adervs{\bar{\psi}}{x} \right) ,
\\
&i \hbar\ \aderv{}{t}\ |\psi|^2 
= -\frac{\hbar^2}{2 m}\ \aderv{}{x} \left( \bar{\psi}\ \aderv{\psi}{x} - \psi\ \aderv{\bar{\psi}}{x} \right) ,
\\
&\aderv{P(x,t)}{t} + \aderv{j(x,t)}{x} = 0 ,
}{40}
where the probability current density, $j(x,t)$, is defined as
\eqn{
j(x,t) = \frac{\hbar}{2 i m} \left[ \bar{\psi}(x,t)\ \aderv{\psi(x,t)}{x} - \psi(x,t)\ \aderv{\bar{\psi}(x,t)}{x} \right] .
}{41}
\equa{40} has the form of a conservation law for the density $P(x,t)$ with the associated current $j(x,t)$. A natural consequence of the wave function $\psi(x,t)$ being square integrable over the real axis, is that
\eqn{
\lim_{x \to \pm \infty} \psi(x,t) = 0 ,
}{42}
and the same obviously holds for its complex conjugate, $\bar{\psi}(x,t)$, so that it must hold for $j(x,t)$, as one can infer from its definition. Then, on integrating \equa{40} with respect to $x$ over the real axis $\mathbb{R}$, we conclude that
\eqn{
\aderv{}{t} \int_{-\infty}^{\infty} dx\ P(x,t) = \lim_{x \to - \infty} j(x,t) - \lim_{x \to  \infty} j(x,t) = 0.
}{43}
\equa{43} proves the time--independence of the integral on the left hand side of the normalization condition \equa{37}, and definitely justifies the choice to set this integral to $1$.

We mention that there exist quantum states such that the probability density $P(x,t) = |\psi(x,t)|^2$ is uniform over the real axis. In these states, it is evidently not possible to normalize the wave function as prescribed in \equa{37}, since the integral of $|\psi(x,t)|^2$ over the range $(-\infty, \infty )$ is infinite. This happens for the quantum states where the momentum of the particle is known with certainty. In these states, as a consequence of the uncertainty principle \equa{28}, the position of the particle is completely unknown, so that the probability density $P(x,t)$ is uniformly distributed over the real axis. The impossibility to prescribe the normalization condition \equa{37} does not affect in any sense the statistical interpretation of the probability density $P(x,t) = |\psi(x,t)|^2$.

\vfill\eject
\section{The Schr\"odinger equation}
The conservation of the probability density, \equa{40}, holds also when the particle is not free, but subject to a potential $V(x)$. However, the wave equation (\ref{31}) must be generalized, so that it is replaced by
\eqn{
i \hbar\ \aderv{\psi(x,t)}{t} = - \frac{\hbar^2}{2 m}\ \adervs{\psi(x,t)}{x} + V(x)\ \psi(x,t).
}{44}
\equa{44} is called the one--dimensional \emph{Schr\"odinger equation}. It can be easily extended to the three-dimensional case,
\eqn{
i \hbar\ \aderv{\psi(\vec{r},t)}{t} = - \frac{\hbar^2}{2 m}\ \nabla^2 \psi(\vec{r},t) + V(\vec{r})\ \psi(\vec{r},t),
}{45}
where $\vec{r} = (x,y,z)$ is the position vector. 
\\
The conservation of the probability density, 
\eqn{
P(\vec{r},t) = |\psi(\vec{r},t)|^2 ,
}{46}
now reads
\eqn{
\aderv{P(\vec{r},t)}{t} + \bnabla \cdot \vec{j}(\vec{r},t) = 0 ,
}{47}
where the probability current density is a vector field given by
\eqn{
\vec{j}(\vec{r},t) = \frac{\hbar}{2 i m} \left[ \bar{\psi}(\vec{r},t)\ \bnabla \psi(\vec{r},t) - \psi(\vec{r},t)\ \bnabla \bar{\psi}(\vec{r},t) \right] .
}{48}

\subsection{Expectation values of physical observables}
Let us consider a physical observable expressed as a function of the position, $f(x)$. Since $P(x, t) = |\psi(x,t)|^2 = \bar{\psi}(x,t)\ \psi(x,t)$ is the probability density, one can evaluate the expectation value of $f(x)$ as
\eqn{
\langle f(x) \rangle = \int_{-\infty}^{\infty} dx\ f(x)\ P(x,t) = \int_{-\infty}^{\infty} dx\ f(x)\ |\psi(x,t)|^2 = \int_{-\infty}^{\infty} dx\ \bar{\psi}(x,t)\ f(x)\ \psi(x,t) ,
}{49}
Let us recall the relationship between position $x$ and momentum $p$ in classical mechanics,
\eqn{
p = m v = m\ \adervo{x}{t} .
}{50}
A reasonable guess is that, for a quantum system, the expectation values of $x$ and $p$ follow the same relationship, 
\eqn{
\langle p \rangle = m\ \adervo{\langle x \rangle}{t} .
}{51}
Hence, on account of \equa{49}, we can write
\spl{
\langle p \rangle = m\ \adervo{\langle x \rangle}{t} &= m\ \adervo{}{t}\int_{-\infty}^{\infty} dx\ \bar{\psi}\ x\ \psi \\
&= m\ \int_{-\infty}^{\infty} dx \left( \aderv{\bar{\psi}}{t}\ x\ \psi + \bar{\psi}\ x\ \aderv{\psi}{t} \right) \\
&= \frac{\hbar}{2 i}\ \int_{-\infty}^{\infty} dx \left( \adervs{\bar{\psi}}{x}\ x\ \psi - \bar{\psi}\ x\ \adervs{\psi}{x} \right)
,
}{52}
where we used the Schr\"odinger equation (\ref{44}). We can employ the integration by parts to rewrite \equa{52} as,
\spl{
\langle p \rangle &= \frac{\hbar}{2 i} \left[ \aderv{\bar{\psi}}{x}\ x\ \psi - \bar{\psi}\ x\ \aderv{\psi}{x} \right]_{-\infty}^{\infty} - \frac{\hbar}{2 i}  \int_{-\infty}^{\infty} dx \ \aderv{\bar{\psi}}{x}\ \psi + \frac{\hbar}{2 i}  \int_{-\infty}^{\infty} dx \ \bar{\psi}\ \aderv{\psi}{x} \\
&= \frac{\hbar}{2 i} \left[ \aderv{\bar{\psi}}{x}\ x\ \psi - \bar{\psi}\ x\ \aderv{\psi}{x} - \bar{\psi}\ \psi \right]_{-\infty}^{\infty} + \frac{\hbar}{i}  \int_{-\infty}^{\infty} dx \ \bar{\psi}\ \aderv{\psi}{x} = \int_{-\infty}^{\infty} dx \ \bar{\psi}\ \frac{\hbar}{i}\  \aderv{}{x}\ \psi
.
}{53}
In fact, for square integrable functions, we have
\eqn{
\left[ \aderv{\bar{\psi}}{x}\ x\ \psi - \bar{\psi}\ x\ \aderv{\psi}{x} - \bar{\psi}\ \psi \right]_{-\infty}^{\infty} = 0 .
}{54}
The expectation value of $p$ evaluated in \equa{53},
\eqn{
\langle p \rangle = \int_{-\infty}^{\infty} dx \ \bar{\psi}(x,t)\ \frac{\hbar}{i}\  \aderv{}{x}\ \psi(x,t) ,
}{55}
strongly suggests that we identify the momentum with the differential operator
\eqn{
p = \frac{\hbar}{i}\ \aderv{}{x} .
}{56}
Then, in general, we have
\eqn{
\langle f(p) \rangle = \int_{-\infty}^{\infty} dx \ \bar{\psi}(x,t)\ f\Bigg(\frac{\hbar}{i}\  \aderv{}{x}\Bigg) \ \psi(x,t) .
}{57}
A specially interesting quantity is the Hamiltonian function of classical mechanics (see \ref{anamec}),
\eqn{
H(x,p) = \frac{p^2}{2 m} + V(x) .
}{58}
On account of \equa{56}, this corresponds to the quantum mechanical \emph{Hamiltonian operator},
\eqn{
H = - \frac{\hbar^2}{2 m}\ \adervs{}{x} + V(x) .
}{59}
We have achieved an important point: the quantum mechanical observables are defined through their expectation values. Thus, they are defined as operators acting on the wave function $\psi(x,t)$. The position operator is simple, just the multiplication of $x$ by $\psi(x,t)$. The momentum operator and the Hamiltonian operator are less trivial as they involve derivatives. The classical observables commute, since $x p$ and $p x$ coincide, but the quantum observables in general do not. For example, we may evaluate the commutator of the quantum operators $x$ and $p$,
\eqn{
[x, p] = x p - p x .
}{60}
We can write
\eqn{
[x, p]\ \psi(x,t) = x\ \frac{\hbar}{i}\ \aderv{}{x}\ \psi(x,t) - \frac{\hbar}{i}\ \aderv{}{x}\left[ x\ \psi(x,t) \right] = - \frac{\hbar}{i}\ \psi(x,t) = i \hbar\ \psi(x,t) ,
}{61}
and this means that
\eqn{
[x, p] = i \hbar .
}{62}

\subsection{The $x$--space and the $p$--space}
A representation of the wave function $\psi(x,t)$ can be based on a Fourier transform
\eqn{
\psi(x,t) = \frac{1}{\sqrt{2 \pi \hbar}} \int_{-\infty}^{\infty} dp\ \phi(p, t) e^{i p x/\hbar} .
}{63}
The inversion theorem of Fourier transforms yields
\eqn{
\phi(p, t) = \frac{1}{\sqrt{2 \pi \hbar}} \int_{-\infty}^{\infty} dx\ \psi(x,t) e^{- i p x/\hbar} .
}{64}
A proof of the inversion theorem is provided in \ref{appenB}.
One can think of $\phi(p, t)$ as the wave function in the momentum space, while $\psi(x,t)$ is the wave function in the coordinate space. For the sake of brevity, the coordinate space is called $x$--space, while the momentum space is called $p$--space. The interpretation of $\phi(p, t)$ as the wave function in $p$--space is supported by the Parseval--Plancherel formula (a proof can be found in \ref{appenB}),
\eqn{
\int_{-\infty}^{\infty} dp\ |\phi(p, t)|^2 = \int_{-\infty}^{\infty} dx\ |\psi(x,t)|^2 . 
}{65}
\equa{65} implies that the normalization of $\psi(x,t)$ coincides with the normalization of $\phi(p, t)$. Then, \equa{37} yields
\eqn{
\int_{-\infty}^{\infty} dp\ |\phi(p, t)|^2 = 1 . 
}{66}
From \equasa{55}{63}, we can write the expectation value of $p$ as
\spl{
\langle p \rangle &= \int_{-\infty}^{\infty} dx \ \bar{\psi}(x,t)\ \frac{\hbar}{i}\  \aderv{}{x}\ \psi(x,t) \\
&= \int_{-\infty}^{\infty} dx \ \bar{\psi}(x,t)\ \frac{\hbar}{i}\  \aderv{}{x}\ \frac{1}{\sqrt{2 \pi \hbar}} \int_{-\infty}^{\infty} dp\ \phi(p, t) e^{i p x/\hbar} \\
&= \int_{-\infty}^{\infty} dx \ \bar{\psi}(x,t)\ \frac{1}{\sqrt{2 \pi \hbar}} \int_{-\infty}^{\infty} dp\ p\ \phi(p, t) e^{i p x/\hbar} \\
&= \int_{-\infty}^{\infty} dp\ p\ \phi(p, t)\ \frac{1}{\sqrt{2 \pi \hbar}} \int_{-\infty}^{\infty} dx\ \bar{\psi}(x,t) e^{i p x/\hbar} \\
&= \int_{-\infty}^{\infty} dp\ p\ \phi(p, t)\ \bar{\phi}(p,t) = \int_{-\infty}^{\infty} dp\ p\ |\phi(p, t)|^2
.
}{67}
\equasa{66}{67} suggest that $|\phi(p, t)|^2\ dp$ can be interpreted as the probability of the particle to have a momentum between $p$ and $p + dp$.

The possibility to define a wave function in the $x$--space, $\psi(x,t)$, and a wave function in the $p$--space, $\phi(p, t)$, leads us to the conclusion that we may have different \emph{representations} of the wave function. If we adopt the $x$--space representation of the wave function, then the quantum mechanical operators for the position and for the momentum are
\eqn{
\hat{x} = x, \qquad \hat{p} = \frac{\hbar}{i}\ \aderv{}{x} . 
}{68}
Hereafter, the quantum mechanical operators will be denoted by the \emph{hat} symbol, $\hat{~}\,$. This is a common practice to distinguish the generally non--commuting quantum observables, or \emph{q--numbers}, from the commuting classical observables, or \emph{c--numbers}.

If we adopt the $p$--space representation of the wave function, then the quantum mechanical operators for the position and for the momentum are
\eqn{
\hat{x} = i \hbar\ \aderv{}{p}, \qquad \hat{p} = p . 
}{69}
The form of the position operator $\hat{x}$ in the $p$--space representation is easily proved on evaluating the expectation value of the position $x$,  
\spl{
\langle x \rangle &= \int_{-\infty}^{\infty} dx \ \bar{\psi}(x,t)\ x\ \psi(x,t) \\
&= \int_{-\infty}^{\infty} dx \ \bar{\psi}(x,t)\ x\ \frac{1}{\sqrt{2 \pi \hbar}} \int_{-\infty}^{\infty} dp\ \phi(p, t) e^{i p x/\hbar} \\
&= \int_{-\infty}^{\infty} dp\ \phi(p, t)\ \frac{1}{\sqrt{2 \pi \hbar}} \int_{-\infty}^{\infty} dx\ x\ \bar{\psi}(x,t) e^{i p x/\hbar} \\
&= - \int_{-\infty}^{\infty} dp\ \phi(p, t)\ i \hbar\ \aderv{}{p}\ \frac{1}{\sqrt{2 \pi \hbar}} \int_{-\infty}^{\infty} dx\ \bar{\psi}(x,t) e^{i p x/\hbar} \\
&= - \int_{-\infty}^{\infty} dp\ \phi(p, t)\ i \hbar\ \aderv{}{p}\ \bar{\phi}(p,t) = \int_{-\infty}^{\infty} dp\ \bar{\phi}(p,t)\ i \hbar\ \aderv{}{p}\ \phi(p,t)
,
}{70}
where we used the obvious fact that $\langle x \rangle$ is real, so that it coincides with its complex conjugate.\\
In general, we have
\eqn{
\langle f(x) \rangle = \int_{-\infty}^{\infty} dp \ \bar{\phi}(p,t)\ f\Bigg(i \hbar\  \aderv{}{p}\Bigg) \ \phi(p,t) .
}{71}

\subsection{On the general structure of quantum mechanics}
We learned from the analysis of the wave--particle duality in a quantum particle that all the accessible information on the state of a particle at a given time can be obtained from the wave function $\psi$, in terms of the expectation values of physical observables. On the other hand, the one--dimensional state of a classical particle is determined by the pair $(x, p)$, position and momentum, at a given time. In classical mechanics, the space of states is the \emph{phase space}, i.e. the vector space of the pairs $(x, p)$.

The space of the states of a quantum particle is the functional space whose elements are the wave functions $\psi$. The space of states is a linear space. This feature is a consequence of the linearity of the Schr\"odinger equation: a linear combination of any two solutions of the Schr\"odinger equation is still a solution of the Schr\"odinger equation. The possibility of having different representations of the quantum states, such as the $x$--space representation and the $p$--space representation discussed in the preceding section, suggests that we denote the states of a quantum system in a representation--independent form. 

Paul Dirac proposed a notation based on the \emph{ket} symbol: a quantum state is expressed as $\ket{\psi}$. The linear space of the states of a quantum system is endowed with a norm and with an inner product,
\gat{
\parallel\! \psi\! \parallel^2 = \braket{\psi}{\psi} = \int_{-\infty}^{\infty} dx \ |\psi(x,t)|^2 ,\\
\braket{\psi}{\varphi} = \int_{-\infty}^{\infty} dx \ \bar{\psi}(x,t) \varphi(x,t) , \qquad
\braket{\varphi}{\psi} = \bar{\braket{\psi}{\varphi}} .
}{72}
The inner product between $\ket{\psi}$ and $\ket{\varphi}$ is obtained by transforming the ket $\ket{\psi}$ into a \emph{bra}, $\bra{\psi}$ and then contracting the bra $\bra{\psi}$ with the ket $\ket{\varphi}$ to obtain the \emph{bra--c--ket} $\braket{\psi}{\varphi}$. Formally the bra is the dual vector of a corresponding ket, the relationship between the bras and the kets is just as that between row vectors and column vectors in a $\mathbb{R}^n$ vector space. The main point is that the vector space of the quantum states $\ket{\psi}$ is not finite--dimensional. The mathematical structure of this vector space is what the mathematicians call a \emph{Hilbert space}. In practice, a Hilbert space is an infinite--dimensional vector space endowed with an inner product. We will denote as $\mathcal{H}$ the Hilbert space of the quantum states. The main properties of the inner product are
\gat{
\braket{\psi}{\varphi} = \overline{\braket{\varphi}{\psi}} ,\\
\braket{\psi}{\psi} \ge 0, \qquad\quad  \braket{\psi}{\psi} = 0\quad \Longrightarrow\quad \ket{\psi}=0 ,\\
\braket{\psi}{c_1  \varphi_1 + c_2 \varphi_2} = c_1 \braket{\psi}{\varphi_1} + c_2 \braket{\psi}{\varphi_2} , \qquad c_1, c_2 \in \mathbb{C} .
}{73}
If the quantum states are the elements of the Hilbert space $\mathcal{H}$, the quantum observables are linear operators in the vector space $\mathcal{H}$. A physical observable $\hat{A}$ must be a special linear operator. In fact, let $\ket{a}$ be a state such that the observable $\hat{A}$ displays the value $a$, namely
\eqn{
\ket{\hat{A} a} = a \ket{a} .
}{74}
We can collect all the possible values allowed for the observable $\hat{A}$, $\{ a_1, a_2,\ \ldots\ , a_n,\ \ldots\ \}$, and consider the quantum states, $\{ \ket{a_1}, \ket{a_2},\ \ldots\ , \ket{a_n},\ \ldots\ \}$, where $\hat{A}$ displays one of these allowed values. It is straightforward, from \equa{74}, to identify $\{ a_1, a_2,\ \ldots\ , a_n,\ \ldots\ \}$ with the \emph{eigenvalues} of the linear operator $\hat{A}$ and $\{ \ket{a_1}, \ket{a_2},\ \ldots\ , \ket{a_n},\ \ldots\ \}$ with the eigenvectors or, more precisely, with the \emph{eigenstates}. If it is intended to describe a physical observable, the linear operator $\hat{A}$ must have real eigenvalues, i.e. $\hat{A}$ must be a \emph{Hermitian operator}. A Hermitian operator $\hat{A}$ coincides with its adjoint operator $\hat{A}^\dagger$, defined so that
\eqn{
\braket{\psi}{\hat{A} \varphi} = \braket{\hat{A}^\dagger \psi}{\varphi} , \qquad \forall\ \ket{\psi}, \ket{\varphi} \in \mathcal{H} .
}{75}
As a consequence of \equa{75}, a Hermitian operator is also called \emph{self--adjoint}, namely $\hat{A}^\dagger = \hat{A}$.\\
Very useful properties of the adjoint operators regard the sum and the product of two operators $\hat{A}$ and $\hat{B}$, as well as the product by a complex number $\lambda$,
\eqn{
\big( \hat{A} + \hat{B} \big)^\dagger = \hat{A}^\dagger + \hat{B}^\dagger , \qquad
\big( \hat{A} \hat{B} \big)^\dagger = \hat{B}^\dagger \hat{A}^\dagger , \qquad
\big( \lambda \hat{A} \big)^\dagger = \bar{\lambda}\, \hat{A}^\dagger.
}{76}
A quantum observable does not necessarily display a discrete sequence of eigenvalues, the generic eigenvalue $a$ can also vary with continuity. One can think, for instance, to the one--dimensional position operator $\hat{x}$. Evidently, the eigenvalues of $\hat{x}$ are all the possible positions $x$ in the real axis. Depending on the \emph{spectrum} of the eigenvalues being discrete or continuous, we may have different forms of normalization for the eigenstates of an observable $\hat{A}$. Eigenstates corresponding to different eigenvalues are orthogonal, i.e. their inner product is zero. For a continuous spectrum, we have a normalization formula based on \emph{Dirac's delta function},
\eqn{
\braket{a}{a'} = \delta(a - a') ,
}{77}
where $\ket{a}$ and $\ket{a'}$ are any two eigenstates of $\hat{A}$.\\
For a discrete spectrum, we have a normalization formula based on \emph{Kronecker's delta symbol},
\eqn{
\braket{a_n}{a_m} = \delta_{n m} ,
}{78}
where $\ket{a_n}$ and $\ket{a_m}$ are any two eigenstates of $\hat{A}$.

\subsection{Bases in the Hilbert space and representations}
The eigenstates of any Hermitian operator $\hat{A}$ may form an \emph{orthonormal} basis of the Hilbert space $\mathcal{H}$, where the term orthonormal means that the elements of the basis are mutually orthogonal, with respect to the inner product, and that they are normalized, with respect to either \equa{77} or (\ref{78}). Thus, any quantum state $\ket{\psi}$ can be expanded as a linear combination of the eigenstates of $\hat{A}$,
\eqn{
\ket{\psi} = \sum_a \ket{a} \braket{a}{\psi} .
}{79}
The sum over $a$ is to be considered, entirely or partly, as an integral over $a$ if the whole spectrum of the eigenvalues of $\hat{A}$, or part of it, is continuous. From the expansion (\ref{79}), one can define a projector, i.e. an operator $\ket{a}\!\bra{a}$ that yields the projection of $\ket{\psi}$ in the subspace spanned by $\ket{a}$. The completeness of the basis, i.e. the possibility to express \emph{any} state $\ket{\psi} \in \mathcal{H}$ by the expansion (\ref{79}), can be formally rephrased as,
\eqn{
\sum_a \ket{a}\! \bra{a} = \hat{I} ,
}{80}
where $\hat{I}$ is the identity operator in $\mathcal{H}$. One may easily prove that two commuting Hermitian operators $\hat{A}$ and $\hat{B}$ may have common eigenstates,
\spl{
&[\hat{A}, \hat{B}] = 0 , \qquad \ket{\hat{A} a} = a \ket{a},\\
&\ket{\hat{B}\hat{A} a} = a \ket{\hat{B} a},\\
&\ket{\hat{A}\hat{B} a} = a \ket{\hat{B} a}.
}{81}
Therefore, $\ket{\hat{B} a}$ is an eigenstate of $\hat{A}$ corresponding to the eigenvalue $a$. This may imply two possible cases.
\begin{enumerate}
\item The eigenvalue $a$ is non--degenerate, so that $\ket{\hat{B} a} = b \ket{a}$, where $b \in \mathbb{R}$, namely $\ket{a}$ is an eigenstate of $\hat{B}$. 
\item The eigenvalue $a$ is degenerate, so that $\ket{\hat{B} a}$ can be expressed as a linear combination of different eigenstates $\big\{ \ket{a}, \ket{a'}, \ket{a''},\ \ldots\ \big\}$ of the operator $\hat{A}$ corresponding to the same eigenvalue $a$. In this case, the operator $\hat{B}$, as well as every other Hermitian operator that commutes with $\hat{A}$, may serve to remove the degeneracy.
\end{enumerate}
If we consider the orthonormal basis formed with the eigenstates $\big\{ \ket{a} \big\}$ of a given observable $\hat{A}$, we obtain the $a$--space representation of the quantum state $\ket{\psi(t)}$ at time $t$,
\eqn{
\psi(a,t) = \braket{a}{\psi(t)}  .
}{82}
The complex--valued function $\psi(a,t)$ is the wave function corresponding to the quantum state $\ket{\psi(t)}$ in the $a$--space representation. Examples are the $x$--space representation and the $p$--space representation discussed in the preceding section, obtained when $\hat{A}$ is the position operator $\hat{x}$ or the momentum operator $\hat{p}$, respectively,
\spl{
\psi(x,t) = \braket{x}{\psi(t)}  ,\\
\psi(p,t) = \braket{p}{\psi(t)}  .
}{83}
At an instant $t$, the expectation value of an observable $\hat{A}$ in a quantum state $\ket{\psi(t)}$ is expressed, in the Dirac notation, as
\eqn{
\langle \hat{A} \rangle = \braket{\psi(t)}{\hat{A} \psi(t)},
}{84}
so that, on account of \equa{80}, we can write
\spl{
\langle \hat{A} \rangle = \sum_a \braket{\psi(t)}{a} \braket{a}{\hat{A} \psi(t)}
&= \sum_a \braket{\psi(t)}{a} \braket{\hat{A} a}{\psi(t)}
= \sum_a a\ \braket{\psi(t)}{a} \braket{a}{\psi(t)}\\
&= \sum_a a\ \bar{\psi(a,t)}\ \psi(a,t) = \sum_a a\ |\psi(a,t)|^2.
}{85}
Again, the sum over $a$ can be possibly replaced by an integral when the eigenvalue $a$ is a continuous variable.

\subsection{The statistical interpretation of quantum states}

On the basis of \equa{85}, we can formulate the statistical interpretation of the wave function $\psi(a,t)$. If the spectrum of the eigenvalues $\{ a \}$ is continuous, then $|\psi(a, t)|^2 da$ is the probability that a measurement of the observable $\hat{A}$ yields an outcome within $a$ and $a + d a$. If the spectrum of the eigenvalues $\{ a \}$ is discrete, then $|\psi(a, t)|^2$ is the probability that a measurement of the observable $\hat{A}$ yields the outcome $a$. An immediate consequence of this statistical interpretation is the following. If $\ket{\psi(t)}$ coincides with the eigenstate $\ket{a}$ at some instant $t$, then a measurement of the observable $\hat{A}$ at that instant yields \emph{with certainty}, i.e. with probability $1$, the outcome $a$. We said that two commuting observables $\hat{A}$ and $\hat{B}$ may have common eigenstates. Let $\ket{a}$ be such an eigenstate, displaying the eigenvalue $a$ of $\hat{A}$ and the eigenvalue $b$ of $\hat{B}$. Then, in the state $\ket{a}$, simultaneous measurements of the observables $\hat{A}$ and $\hat{B}$ yield \emph{with certainty}, i.e. with probability $1$, the values $a$ and $b$ respectively.

\subsection{Operator methods and the uncertainty principle}
The use of an abstract Hilbert space representation of the quantum states allows one to obtain  a rigorous proof of the uncertainty principle. Let us consider two Hermitian operators, $\hat{A}$ and $\hat{B}$, in the Hilbert space $\mathcal{H}$. We assume that $\hat{A}$ and $\hat{B}$ are non--commuting. The following lemma holds.
\begin{lem}\label{lem1}
Let $\hat{C}$ be a linear operator in $\mathcal{H}$, such that
$$ [ \hat{A}, \hat{B} ] = i \hat{C} , $$
where $\hat{A}$ and $\hat{B}$ are two Hermitian operators in $\mathcal{H}$. Then, $\hat{C}$ is also Hermitian.
\end{lem}
\begin{proof}
On account of \equa{76}, we may write $\big( i \hat{C} \big)^\dagger = - i \hat{C}^\dagger$. Moreover,
$$ [ \hat{A}, \hat{B} ]^\dagger = \big( \hat{A} \hat{B} \big)^\dagger - \big( \hat{B} \hat{A} \big)^\dagger = \hat{B}^\dagger \hat{A}^\dagger - \hat{A}^\dagger \hat{B}^\dagger  = \hat{B} \hat{A} - \hat{A} \hat{B} = [ \hat{B}, \hat{A} ] = - i \hat{C}.
$$
This proves that $\hat{C} = \hat{C}^\dagger$.
\end{proof}
We define the uncertainties associated with the observables $\hat{A}$ and $\hat{B}$ in terms of the square deviations from their average values,
\gat{
\Delta A^2 = \big< \big( \hat{A} - \langle \hat{A} \rangle \big)^2 \big> = \braket{\psi}{\big( \hat{A} - \langle \hat{A} \rangle \big)^2 \psi},  \\
\Delta B^2 = \big< \big( \hat{B} - \langle \hat{B} \rangle \big)^2 \big> = \braket{\psi}{\big( \hat{B} - \langle \hat{B} \rangle \big)^2 \psi}. 
}{86}
We can define the operators
\eqn{
\hat{V} = \hat{A} - \langle \hat{A} \rangle , \qquad \hat{W} = \hat{B} - \langle \hat{B} \rangle,
}{87}
so that 
\eqn{
\Delta A^2 = \braket{\psi}{\hat{V}^2 \psi}, \qquad \Delta B^2 = \braket{\psi}{\hat{W}^2 \psi} .
}{88}
We note that
\eqn{
[ \hat{V}, \hat{W} ] = \big( \hat{A} - \langle \hat{A} \rangle \big) \big( \hat{B} - \langle \hat{B} \rangle \big) - \big( \hat{B} - \langle \hat{B} \rangle \big) \big( \hat{A} - \langle \hat{A} \rangle \big) 
= [ \hat{A}, \hat{B} ] .
}{89}
We now evaluate the scalar product of the state
\[
\ket{\hat{V} \psi} + i \lambda \ket{\hat{W} \psi} , \qquad \lambda \in \mathbb{R},
\]
with itself, namely
\spl{
&\braket{\hat{V} \psi}{\hat{V} \psi} + \lambda^2 \braket{\hat{W} \psi}{\hat{W} \psi} + i \lambda \braket{\hat{V} \psi}{\hat{W} \psi} - i \lambda \braket{\hat{W} \psi}{\hat{V} \psi} \ge 0,\\
&\braket{\psi}{\hat{V}^2 \psi} + \lambda^2 \braket{\psi}{\hat{W}^2 \psi} + i \lambda \braket{\psi}{[ \hat{V}, \hat{W} ] \psi} \ge 0,\\
&\Delta A^2 + \lambda^2 \Delta B^2 + i \lambda \braket{\psi}{[ \hat{A}, \hat{B} ] \psi} \ge 0.
}{90}
The inequality (\ref{90}) must hold for every $\lambda \in \mathbb{R}$. The choice of $\lambda$ that ensures the most restrictive form of the inequality (\ref{90}) is obtained by seeking the minimum of the left hand side of \equa{90} with respect to $\lambda$. This implies that we force the derivative of 
\[
\Delta A^2 + \lambda^2 \Delta B^2 + i \lambda \braket{\psi}{[ \hat{A}, \hat{B} ] \psi}
\]
with respect to $\lambda$ to be zero. Thus, we obtain
\eqn{
\lambda = - \frac{i \braket{\psi}{[ \hat{A}, \hat{B} ] \psi}}{2\, \Delta B^2} ,
}{91}
so that, \equa{90} yields
\spl{
&\Delta A^2 - \frac{\braket{\psi}{[ \hat{A}, \hat{B} ] \psi}^2}{4\, \Delta B^2} + \frac{\braket{\psi}{[ \hat{A}, \hat{B} ] \psi}^2}{2\, \Delta B^2} \ge 0,\\
&\Delta A^2\, \Delta B^2 \ge - \frac{1}{4}\ \braket{\psi}{[ \hat{A}, \hat{B} ] \psi}^2 .
}{92}
As a consequence of Lemma~\ref{lem1}, the right hand side of \equa{92} is a non--negative real number.

\equa{92} is a generalized form of Heisenberg's uncertainty principle. An important application is for the pair of observables $\hat{x}$ and $\hat{p}$. We know, in fact, from \equa{62} that   
\eqn{
[\hat{x},\hat{p}] = i \hbar ,
}{93}
so that the inequality (\ref{92}) yields
\eqn{
\Delta x^2\, \Delta p^2 \ge \frac{1}{4}\ \hbar^2 ,
}{94}
or equivalently
\eqn{
\Delta x\, \Delta p \ge \frac{1}{2}\ \hbar ,
}{95}
\equa{92} reveals an important fact about the uncertainty principle: it is not possible to achieve simultaneously and \emph{with certainty} the values of two \emph{non--commuting} observables. On the other hand, one can achieve simultaneously and \emph{with certainty} the values of two \emph{commuting} observables. The second fact is a consequence of the possibility to determine a basis of common eigenstates of two commuting observables for the Hilbert space $\mathcal{H}$.

\subsection{The Hamiltonian operator, the evolution operator, and the\\ stationary states}
The expression of the Hermitian operator that describes a quantum observable is obtained from the corresponding expression of classical mechanics. For instance, as already pointed out in \equasa{58}{59}, the Hamiltonian is given by
\eqn{
\hat{H} = \frac{\hat{p}^2}{2 m} + V(\hat{x}) ,
}{96}
and the Schr\"odinger equation is expressed as
\eqn{
i \hbar \aderv{}{t}\ \ket{\psi} = \ket{ \hat{H} \psi} .
}{97}
On integrating \equa{97}, one obtains
\eqn{
\ket{\psi(t)} = \ket{ e^{ - i (t - t_0) \hat{H}/ \hbar } \psi(t_0)} ,
}{98}
where the exponential of an operator $\hat{A}$ is defined through its series expansion
\eqn{
e^{ \hat{A}} = \sum_{n=0}^\infty \frac{\hat{A}^n}{n!} .
}{99}
\equa{98} implies that one can define an \emph{evolution operator},
\eqn{
\hat{U}(t, t_0) = e^{ - i (t - t_0) \hat{H}/ \hbar  } .
}{100}
The operator $\hat{U}(t, t_0)$ is not a Hermitian operator, so that it cannot be associated to an observable quantity. On the other hand, $\hat{U}(t, t_0)$ is an \emph{unitary operator}, i.e. it satisfies the mathematical identity
\eqn{
\hat{U}(t, t_0)\ \hat{U}(t, t_0)^\dagger = \hat{U}(t, t_0)^\dagger\ \hat{U}(t, t_0) = \hat{I}  .
}{101}
The eigenstates of the evolution operator $\hat{U}(t, t_0)$ coincide with the eigenstates of $\hat{H}$, as it can be inferred from its definition. The eigenstates of $\hat{H}$, denoted as $\{ \ket{\varepsilon} \}$, are also called the \emph{energy eigenstates},
\eqn{
\ket{\hat{H} \varepsilon} = E \ket{\varepsilon} ,
}{102}
where $E$ is the eigenvalue of $\hat{H}$ in the eigenstate $\ket{\varepsilon}$. \equa{102} plays a special role in the study of quantum systems and it is well--known as the \emph{time--independent Schr\"odinger equation}.\\
The time evolution of an energy eigenstate is obtained from \equa{98}
\eqn{
\ket{\varepsilon(t)} = \ket{ e^{ - i (t - t_0) \hat{H}/ \hbar  } \varepsilon(t_0)} = e^{ - i (t - t_0) E/ \hbar  } \ket{ \varepsilon(t_0)} ,
}{103}
The energy eigenstates are quite peculiar as the expectation value of an arbitrary observable $\hat{A}$ does not depend on time, as it can be deduced from \equa{103},
\eqn{
\braket{\varepsilon(t)}{\hat{A} \varepsilon(t)} = \braket{\varepsilon(t_0)}{\hat{A} \varepsilon(t_0)} , \qquad \forall t, t_0 \in \mathbb{R}.
}{104}
In this sense, the energy eigenstates are termed \emph{stationary states}; they are time--independent except for a phase factor, 
\[
e^{ - i (t - t_0) E/ \hbar  } ,
\]
ineffective with respect to the expectation values of any observable.

\vfill\eject

\section{One--dimensional quantum mechanics}
Simple examples of the quantum mechanical behaviour can be given by considering a particle in a  one--dimensional potential field expressed through a given function $V(x)$. The $x$--space representation is adopted, so that the quantum state is represented by means of the wave function $\psi(x,t)$. The time--independent Schr\"odinger equation for the stationary states
\eqn{
\psi(x,t) = e^{ - i t E/ \hbar  }\ \Psi(x) ,
}{105}
can be written as
\eqn{
- \frac{\hbar^2}{2 m}\ \adervso{\Psi(x)}{x} + V(x) \Psi(x) = E \Psi(x) .
}{106}

\subsection{An infinite potential well}
Let us consider a potential given by,
\eqn{
V(x) = 
\begin{cases}
\infty, \quad &x < 0, \\
0, &0 < x < a , \\
\infty, &x > a . 
\end{cases}
}{107}
The infinite value of $V(x)$ outside the interval $0 < x < a$ means that the space outside this interval is not accessible to the particle. One can think of a particle in a box having width $a$.
\begin{figure}[h!]
\begin{center}
\includegraphics[height=0.37\textwidth]{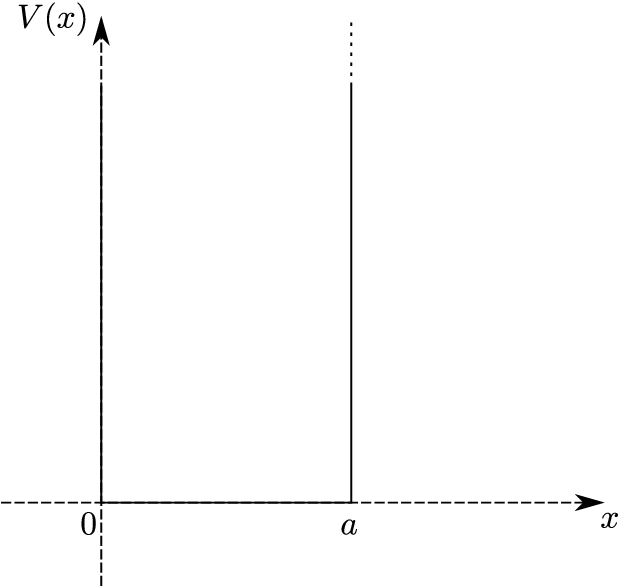}
\end{center}
\caption{The infinite potential well}
\label{fig3}
\end{figure}
Since the particle is not allowed to escape the region $0 < x < a$, the continuity of the probability distribution at $x=0$ and $x=a$ imposes the wave function to vanish at these positions. Thus, the time--independent Schr\"odinger equation, \equa{106}, is subject to the boundary conditions,
\eqn{
\Psi(0) = 0 = \Psi(a) .
}{108}
For $0 < x < a$, \equa{106} can be written as
\eqn{
\adervso{\Psi(x)}{x} = - \frac{2 m E}{\hbar^2}\ \Psi(x) .
}{109}
\equa{109}, with the boundary conditions \equa{108}, admits the non--trivial solutions, i.e. the eigenfunctions,
\eqn{
\Psi_n(x) = C_n \sin(n \pi x/a), \qquad n = 1, 2, 3,\ \ldots\ ,
}{110}
provided that the energy $E$ assumes the eigenvalues,
\eqn{
E_n = \frac{n^2 \pi^2 \hbar^2}{2 m a^2} , \qquad n = 1, 2, 3,\ \ldots\ .
}{111}
\equa{111} reveals an important fact about the quantum particle: its energy cannot assume arbitrary values, but just those defined by the discrete spectrum of the eigenvalues $E_n$.

As it usually happens in the solution of an eigenvalue problem, the normalization constants appearing in \equa{110} can be chosen arbitrarily. However, we should not forget the physical meaning of the wave function $\psi(x,t)$. In fact, the probability distribution $|\psi(x,t)|^2$ must be such that its integral over the interval $0 < x < a$ is equal to $1$,
\spl{
\int_0^a dx\ |\psi(x,t)|^2 = 1, \qquad \Longrightarrow \qquad \int_0^a dx\ |\Psi_n(x)|^2 = 1 ,
}{112}
and this yields $C_n = \sqrt{2/a}$, so that
\eqn{
\Psi_n(x) = \sqrt{\frac{2}{a}}\ \sin(n \pi x/a), \qquad n = 1, 2, 3,\ \ldots\ ,
}{113}
A notable feature of the energy spectrum, \equa{111}, for the particle in the infinite potential well is that the relative increment between two neighbouring energy levels is a decreasing function of $n$,
\eqn{
\frac{\Delta E_n}{E_n} = \frac{E_{n+1} - E_n}{E_n} = \frac{2}{n} + \frac{1}{n^2} , \qquad \lim_{n \to \infty} \frac{\Delta E_n}{E_n} = 0 .
}{114}
This result can be rephrased as follows: when the energy of the particle is such that 
\[
E \gg \frac{\pi^2 \hbar^2}{2 m a^2} , 
\]
the energy spectrum becomes continuous. In other words, when the energy assumes sufficiently high values, the quantum behaviour tends to mimic the classical behaviour of a particle in a box.

Another notable feature of the energy spectrum, \equa{111}, is the existence of a non--vanishing minimum value for the energy of the particle. It is the value of the energy for the level $n=1$. This value,
\eqn{
E_1 = \frac{\pi^2 \hbar^2}{2 m a^2} ,
}{115}
is usually called the \emph{zero--point energy}. The quantum particle cannot display an energy lower than the zero--point energy.

\subsection{Quantum tunnelling through a barrier}
I am not aware of anybody or anything in the macroscopic world that is able to jump over a fence without a kinetic energy sufficient to overcome the gravitational potential energy associated to the height of the fence. A quantum particle can do it: this effect is called \emph{quantum tunnelling}.

\noindent To illustrate the quantum tunnelling, let us consider the potential barrier defined by the function
\eqn{
V(x) = 
\begin{cases}
0, \quad &x < 0, \\
V_0, &0 < x < a , \\
0, &x > a , 
\end{cases}
}{116}
where $V_0 > 0$ is a given constant.
\begin{figure}[h!]
\begin{center}
\includegraphics[height=0.37\textwidth]{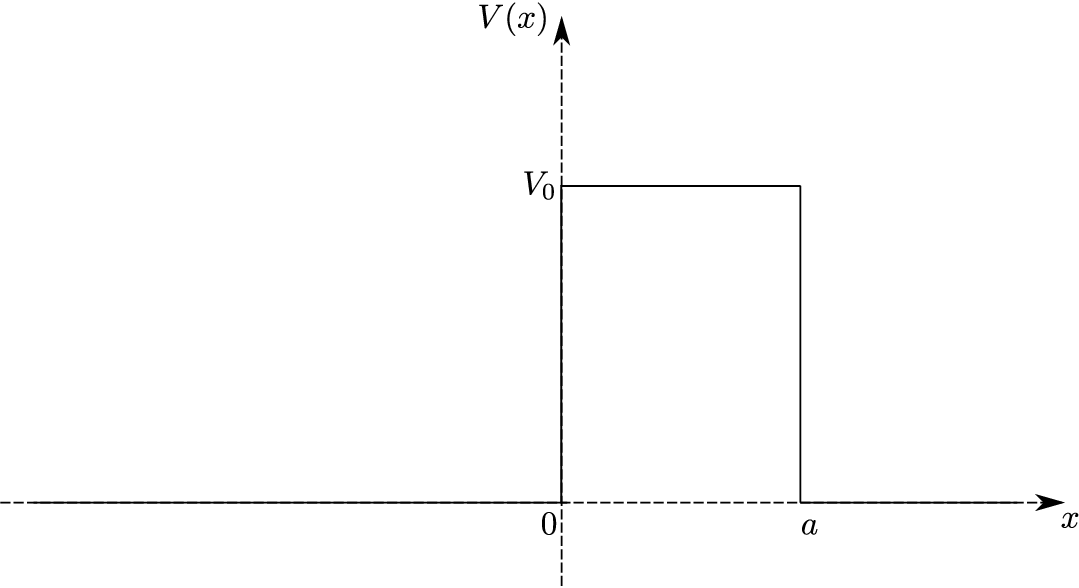}
\end{center}
\caption{The potential barrier}
\label{fig4}
\end{figure}
Unlike the preceding example, the region of space allowed to the quantum particle is in this case the whole real axis. For a classical particle, it is impossible to penetrate the potential barrier if its energy is lower than the barrier height $V_0$. For a quantum particle, things are different as we will see.

On the left of the barrier, $x < 0$, the time--independent Schr\"odinger equation yields
\eqn{
\adervso{\Psi(x)}{x} = - k^2 \Psi(x) , \qquad k = \sqrt{\frac{2 m E}{\hbar^2}}.
}{117}
The general solution is the superposition of incident and reflected travelling waves
\eqn{
\Psi(x) = e^{i k x} + R\ e^{-i k x} .
}{118}
In the barrier, $0 < x < a$, the time--independent Schr\"odinger equation is written as
\eqn{
\adervso{\Psi(x)}{x} = - k_B^2 \Psi(x) , \qquad k_B = \sqrt{\frac{2 m (E - V_0)}{\hbar^2}}.
}{119}
If $E > V_0$, the general solution is 
\eqn{
\Psi(x) = C\ e^{i k_B x} + C'\ e^{-i k_B x} .
}{120}
If $E < V_0$, the general solution is 
\eqn{
\Psi(x) = C\ e^{\beta x} + C'\ e^{-\beta x} , \qquad \beta  = \sqrt{\frac{2 m (V_0 - E)}{\hbar^2}}.
}{121}
On the right of the barrier, $x > a$, the time--independent Schr\"odinger equation yields
\eqn{
\adervso{\Psi(x)}{x} = - k^2 \Psi(x) , \qquad k = \sqrt{\frac{2 m E}{\hbar^2}}.
}{122}
We may have, in this region, a transmitted wave given by
\eqn{
\Psi(x) = T\ e^{i k x} .
}{123}
We can obtain the reflection coefficient $R$, the transmission coefficient $T$, and the integration constants $C$ and $C'$, by imposing the smooth matching of the wave function at the positions $x=0$ and $x=a$, namely
\spl{
&\Psi(0^-) = \Psi(0^+), \qquad \left. \adervo{\Psi}{x} \right|_{x=0^-} = \left. \adervo{\Psi}{x} \right|_{x=0^+} ,\\
&\Psi(a^-) = \Psi(a^+), \qquad \left. \adervo{\Psi}{x} \right|_{x=a^-} = \left. \adervo{\Psi}{x} \right|_{x=a^+} .
}{124}
The continuity of $\Psi(x)$ is an obvious consequence of the continuity of the probability density $|\Psi(x)|^2$. On the other hand, the continuity of the first derivative $d \Psi(x)/ d x$ is a more subtle matter. If one integrates \equa{106} across a point of discontinuity of the potential $V(x)$, say $x=0$, one obtains
\eqn{
- \frac{\hbar^2}{2 m} \left[ \adervo{\Psi(x)}{x} \right]_{-\delta}^\delta = \int_{-\delta}^\delta dx \left[ E - V(x) \right] \Psi(x) , \qquad \delta > 0 .
}{125}
If we let $\delta \to 0$, then the right hand side of \equa{125} must vanish unless we have a Dirac's delta behaviour of $V(x)$ at $x=0$, which is not the present case. Thus, we immediately obtain the continuity of $d \Psi(x)/ d x$ across the point $x=0$. Obviously, the same holds at $x=a$.

We now focus our attention on the case $E < V_0$, i.e. the case such that the penetration of the potential barrier would be impossible for a classical particle. We can determine the constants $R$, $T$, $C$ and $C'$ by employing \equas{118}, (\ref{121}), (\ref{123}), and (\ref{124}),
\spl{
&R = \frac{\left(k^2+\beta ^2\right) \sinh(a \beta)}{2 i k \beta  \cosh(a \beta) + (k^2 - \beta^2) \sinh(a \beta)},\\
&C = \frac{2 k (k-i \beta )}{(k-i \beta )^2-e^{2 a \beta } (k+i \beta )^2},\\
&C' = \frac{2 e^{2 a \beta } k (k+i \beta )}{e^{2 a \beta } (k+i \beta )^2 - (k-i \beta )^2},\\
&T = \frac{4 i e^{a (-i k+\beta )} k \beta }{e^{2 a \beta } (k+i \beta )^2 - (k-i \beta )^2}.
}{126}
Physically meaningful quantities are $|R|^2$ and $|T|^2$, as they express the probability that the particle be reflected or transmitted from the potential barrier, respectively. We obtain
\spl{
&|R|^2 = \frac{(k^2 + \beta ^2)^2 \sinh^2(a \beta)}{4 k^2 \beta ^2 \cosh^2(a \beta) + \left(k^2 - \beta ^2\right)^2 \sinh^2(a \beta)} = \frac{(k^2 + \beta ^2)^2 \sinh^2(a \beta)}{\left(k^2 + \beta ^2\right)^2 \sinh^2(a \beta) + 4 k^2 \beta ^2} ,\\
&|T|^2 = \frac{8 k^2 \beta ^2}{\left(k^2+\beta ^2\right)^2 \cosh(2 a \beta) - k^4 + 6 k^2 \beta ^2 - \beta ^4} 
= \frac{4 k^2 \beta ^2}{\left(k^2 + \beta ^2\right)^2 \sinh^2(a \beta) + 4 k^2 \beta ^2} .
}{127}
From \equa{127}, one can easily see that
\eqn{
|R|^2 + |T|^2 = 1,
}{128}
with the obvious meaning that the particle can be either reflected or transmitted from the barrier. The important point is that $|T|^2 > 0$. This means that, even if $E < V_0$, the particle has a non--zero probability of penetrating the potential barrier. This phenomenon of quantum tunnelling through a potential barrier is quite important. For instance, the thermonuclear fusion of protons in the stars is possible, in most cases, only as a consequence of the quantum tunnelling of protons through the repulsive Coulomb potential barrier. We note that $|T|^2$ decreases with $a$. This is an expected fact, as the larger is the potential barrier the lower is the probability of quantum tunnelling through it.

\subsection{The quantum harmonic oscillator}
As is well--known from classical mechanics, the elastic potential of a simple harmonic oscillator is expressed as
\eqn{
V(x) = \frac{1}{2}\ m \omega^2 x^2 ,
}{129}
where $\omega$ is the angular frequency of the oscillator. Thus, the time--independent Schr\"odinger equation can be written as
\eqn{
- \frac{\hbar^2}{2 m}\ \adervso{\Psi(x)}{x} + \frac{1}{2}\ m \omega^2 x^2\ \Psi(x) = E \Psi(x) .
}{130}
\equa{130} is defined on the whole real axis, and can be rewritten in a more compact form as
\eqn{
- \adervso{\Psi}{q} + q^2 \Psi = \varepsilon \Psi , \qquad \text{with}\quad q = x \sqrt{\frac{m \omega}{\hbar}}, \quad \varepsilon = \frac{2 E}{\hbar \omega} .
}{131}
The conditions at infinity assigned to the solutions of \equa{131} are to make $\Psi$ compatible with the normalization condition \equa{37},
\eqn{
\lim_{x \to \pm \infty} \Psi(x) = 0.
}{132}
We express $\Psi(x)$ as
\eqn{
\Psi = N\ e^{-q^2/2} \varphi(q) ,
}{133}
where $N$ is a normalization constant, so that \equa{131} can be rewritten as
\eqn{
\adervso{\varphi(q)}{q} - 2 q\ \adervo{\varphi(q)}{q} + (\varepsilon - 1) \varphi(q) = 0.
}{134}
\equa{134} is well--known to mathematicians as the \emph{Hermite equation}. It is a second order differential equation with variable coefficients. Its general solution%
\footnote{A good description of this point can be found at the web page:\\[-3pt] \texttt{http://mathworld.wolfram.com/HermiteDifferentialEquation.html}}
can be expressed as a linear combination of the \emph{confluent hypergeometric function of the first kind},
\eqn{
_1\!F_1( - \gamma/4; 1/2; q^2 )    , \qquad \text{with}\quad \gamma = \varepsilon - 1 ,
}{135}
and of the \emph{Hermite function},
\eqn{
H_{\gamma/2} ( q )    .
}{136}
However, we note that the confluent hypergeometric function of the first kind can be approximated for $q^2 \gg 1$ as
\eqn{
_1\!F_1( - \gamma/4; 1/2; q^2 ) \sim \frac{\sqrt{\pi}}{\Gamma(-\gamma/4)} q^{-\gamma /2 - 1} e^{q^2} , 
}{137}
where $\Gamma$ is \emph{Euler's gamma function}. Thus, if $\varphi(q)$ would involve a term proportional to  the confluent hypergeometric function of the first kind, on account of \equasa{133}{137}, the conditions at infinity (\ref{132}) cannot be satisfied. As for the Hermite function, its asymptotic expression when $q \to - \infty$ is expressed as
\eqn{
|H_{\gamma/2} ( q )| \approx  2^{\gamma /2} |q|^{\gamma/2} + e^{q^2} |q|^{- 1 - \gamma/2}\ \frac{ \Gamma(1 + \gamma/2 ) \sin(\pi \gamma / 2)}{\sqrt{\pi}}  .
}{138}
Here, the dangerous term is that proportional to $e^{q^2}$. This term disappears, thus ensuring the validity of the conditions at infinity (\ref{132}), if and only if
\eqn{
\sin(\pi \gamma / 2) = 0, \qquad \Longrightarrow \qquad \gamma = 2 n, \quad  n = 0, 1, 2, 3,\ \ldots\ .
}{139}
Therefore, the conditions at infinity, \equa{132}, are satisfied by considering just the Hermite function with $\gamma$ an even natural number. In this case, the Hermite function is a \emph{Hermite polynomial} with degree equal to $\gamma/2$. Therefore, \equa{133} now reads
\eqn{
\Psi = \Psi_n = N_n\ e^{-q^2/2} H_n ( q ) ,     \qquad \text{with}\quad \varepsilon = 2 n + 1, \quad n = 0, 1, 2, 3,\ \ldots\ . 
}{140}
This means that the allowed energy levels for the harmonic oscillator are given by
\eqn{
E_n = \left( n + \frac{1}{2} \right) \hbar \omega  ,     \qquad n = 0, 1, 2, 3,\ \ldots\ .
}{141}
As in the case of the infinite potential well, we have a zero--point energy of the harmonic oscillator given by the lowest energy level,
\eqn{
E_0 = \frac{1}{2} \ \hbar \omega   .
}{142}
The first 10 Hermite polynomials are
\spl{
 H_0 ( q ) &= 1 ,\\
 H_1 ( q ) &= 2 q ,\\
 H_2 ( q ) &= -2+4 q^2 ,\\
 H_3 ( q ) &= -12 q+8 q^3 ,\\
 H_4 ( q ) &= 12-48 q^2+16 q^4 ,\\
 H_5 ( q ) &= 120 q-160 q^3+32 q^5 ,\\
 H_6 ( q ) &= -120+720 q^2-480 q^4+64 q^6 ,\\
 H_7 ( q ) &= -1680 q+3360 q^3-1344 q^5+128 q^7 ,\\
 H_8 ( q ) &= 1680-13440 q^2+13440 q^4-3584 q^6+256 q^8 ,\\
 H_9 ( q ) &= 30240 q-80640 q^3+48384 q^5-9216 q^7+512 q^9 .\\
}{143}
\equa{143} shows that the Hermite polynomials are even functions of $q$ if $n$ is even, while they are odd functions of $q$ if $n$ is odd. The normalization constants $N_n$ are found so that the normalization condition, \equa{37}, is satisfied
\eqn{
N_n^2 = \frac{1}{2^n n!}\ \sqrt{\frac{m \omega}{\pi\hbar}} .
}{144}
On account of \equas{140}, (\ref{143}), and (\ref{144}), the \emph{ground state} corresponding to the zero--point energy, \equa{142}, is described by a Gaussian wave function,
\eqn{
\Psi_0 = \left({\frac{m \omega}{\hbar \pi}}\right)^{1/4} e^{-q^2/2}  . 
}{145}
The lowest three stationary states are represented qualitatively in Fig.~\ref{fig5}.
\begin{figure}[h!]
\begin{center}
\includegraphics[width=0.8\textwidth]{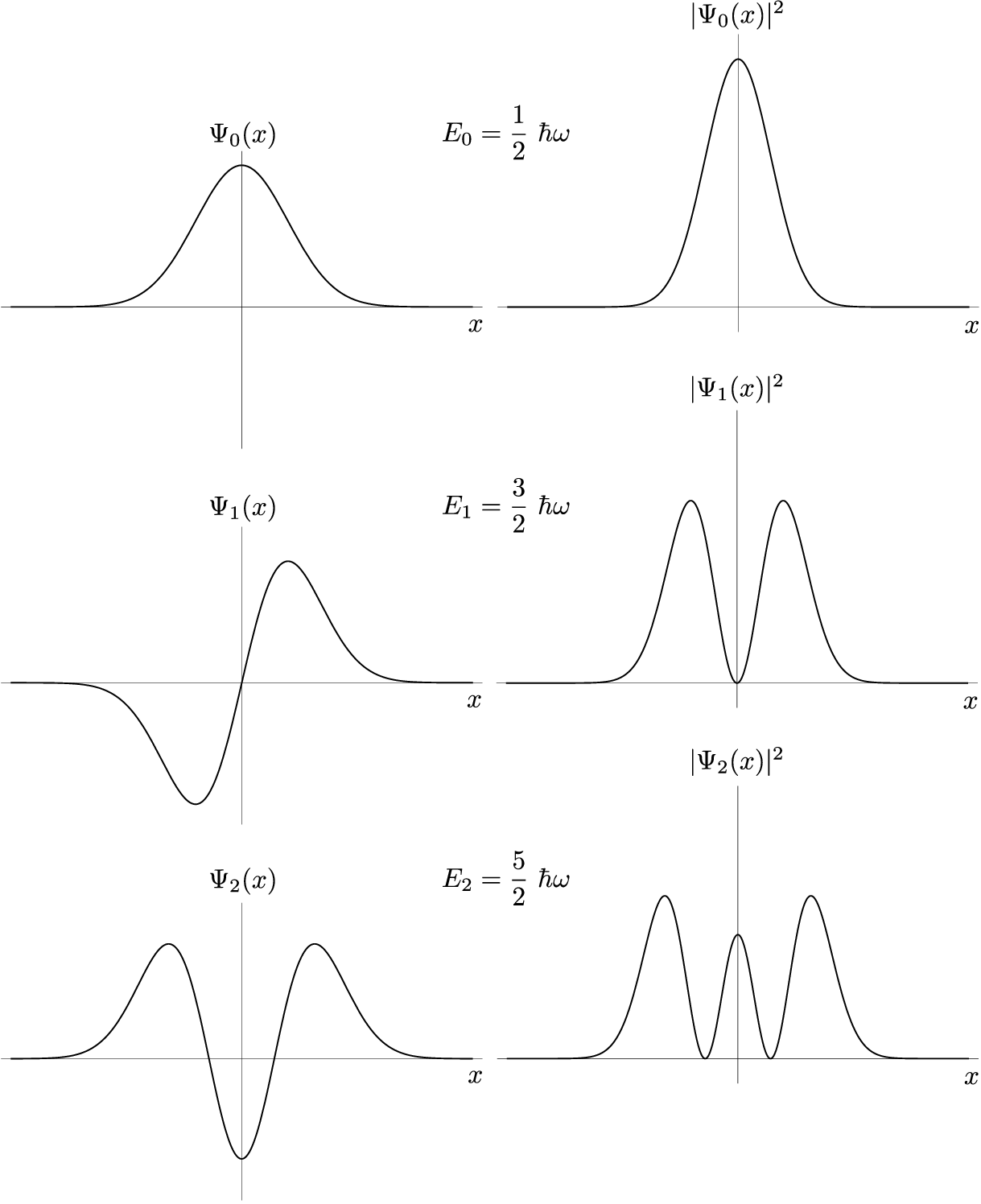}
\end{center}
\caption{Stationary states of the harmonic oscillator with $n = 0, 1, 2$}
\label{fig5}
\end{figure}

An important fact about the quantum harmonic oscillator is that the particle can be everywhere on the real axis, with decreasing probability for large values of $|x|$. On the other hand, the classical harmonic oscillator is confined in a region $-A \le x \le A$, where $A$ is the amplitude of the oscillation. The amplitude $A$ is classically related to the energy $E$ of the oscillator,
\[
E = \frac{1}{2}\ m \omega^2 A^2 ,
\]
so that the classical particle with energy $E$ is confined within the region 
\[
 - \sqrt{\frac{2 E}{m \omega^2}} \le x \le \sqrt{\frac{2 E}{m \omega^2}} .
\]
Thus, the quantum particle is allowed to escape the region where it is classically confined, for a given value of the energy. This phenomenon can be traced back to the quantum tunnel effect, discussed in the preceding section.

\vfill\eject

\section*{Suggestions for further reading}
Among the several treatises on quantum mechanics, the following are strongly suggested for a deeper insight into the behaviour of quantum systems. All these textbooks are at an adequate level for a reader with a basic mathematical education on calculus, Fourier series and integral transforms.

\subsubsection*{\citeauthor{Dirac1981}, \emph{The Principles of Quantum Mechanics}, \citeyear{Dirac1981}} 
It is a classical textbooks whose first edition dates back to 1930. In this book the formal structure of quantum mechanics based on the properties of the vectors (\emph{bras} and \emph{kets}) and the linear operators in a Hilbert space is clearly described.
 
\subsubsection*{\citeauthor{Messiah1999}, \emph{Quantum Mechanics}, \citeyear{Messiah1999}} 
Published in 1958, this book was originally in two volumes. It covers several topics from the basic description of quantum systems to its extension in the domain of relativistic quantum mechanics.

\subsubsection*{\citeauthor{Gasiorowicz1974}, \emph{Quantum Physics}, \citeyear{Gasiorowicz1974}}
The edition of this classical textbook that is cited in these notes is of 1974, but there exists the third edition published in 2003 with substantial changes.

\subsubsection*{\citeauthor{Peres1995}, \emph{Quantum Theory: Concepts and Methods}, \citeyear{Peres1995}}
In this textbook, the focus is on the conceptual foundations of quantum mechanics and of its statistical interpretation. Topics such as the Einstein--Podolsky--Rosen (EPR) paradox and the fundamental questions about the measurement process on a quantum system are extensively discussed.

\subsubsection*{\citeauthor{Sakurai1994}, \emph{Modern Quantum Mechanics}, \citeyear{Sakurai1994}}
This book is a modern classic. Dirac's approach to the description of quantum mechanics is adopted from scratch. It is a textbook at an advanced level with applications of the group theory to the study of angular momentum, and of the symmetries in quantum systems. 

\subsubsection*{\citeauthor{Phillips2003}, \emph{Introduction to Quantum Mechanics}, \citeyear{Phillips2003}}
A recent textbook with several qualities. It is very clear, simple and rigorous, definitely suggested to people looking for a good introductory treatise on quantum physics.

\subsubsection*{\citeauthor{Susskind2015}, \emph{Quantum Mechanics: The Theoretical Minimum}, \citeyear{Susskind2015}}
It is definitely not a conventional textbook on quantum mechanics, but it cannot be classified as a divulgative presentation of the topic as it contains a lot of mathematics. The discussion of quantum theory is made with an original approach which offers a rigorous outline of the main features of the theory in a modern and stimulating manner.

\vfill\eject

\appendix
\renewcommand\thesection{Appendix \Alph{section}}
\setcounter{equation}{0}
\renewcommand{\theequation}{A\arabic{equation}}

\section{}\label{anamec}
\textbf{A quick glance into analytical mechanics}
\\
Let us recall the basic facts about the analytical mechanics of a classical system.

Every mechanical system with $N$ degrees of freedom is described by $N$ \emph{generalized coordinates} $\underaccent{\tilde}{q} = (q_1, q_2,\ \ldots\ , q_N)$, with associated \emph{generalized velocities} $\underaccent{\tilde}{\dot{q}}=(\dot{q}_1, \dot{q}_2,\ \ldots\ , \dot{q}_N)$. Here, the dot stands for time derivative. 
The equations of motion for the mechanical system are formulated in terms of the \emph{Lagrangian function},
\eqn{
L(\underaccent{\tilde}{q}, \underaccent{\tilde}{\dot{q}})  = T - V ,
}{a1}
where $T$ denotes the kinetic energy, and $V$ the potential energy. In \equa{a1}, the Lagrangian $L$ is intended as a function of the generalized coordinates $q_i$ and of the generalized velocities $\dot{q}_i$, while it is considered as independent of time. The time independence of $L$ is appropriate for the \emph{conservative} mechanical systems, i.e. systems that are either isolated or subject to an external stationary force field.

The motion of the mechanical system is described by the \emph{Euler--Lagrange equations},
\eqn{
\frac {d}{d t} \left ( \frac {\partial L}{\partial \dot{q}_i} \right ) -  \frac {\partial L}{\partial q_i} = 0 , \qquad i = 1,2,\ \ldots\ , N .
}{a2}
The \emph{generalized momentum}, $p_i$, associated with the generalized coordinate $q_i$ is defined as
\eqn{
p_i = \aderv{L}{\dot{q}_i} .
}{a3}
A widely used (alternative) terminology for $p_i$ is \emph{momentum canonically conjugated} to $q_i$.
Obviously, one can define $N$ generalized momenta, $\underaccent{\tilde}{p} = (p_1, p_2,\ \ldots\ , p_N)$.

The \emph{Hamiltonian function} is defined through the Legendre transformation of the Lagrangian function, namely
\eqn{
H = \sum_{i=1}^N p_i \dot{q}_i - L .
}{a4}
If $L$ depends only on $\underaccent{\tilde}{q}$ and $\underaccent{\tilde}{\dot{q}}$, one may presume that $H$ is a function of $\underaccent{\tilde}{q}$, $\underaccent{\tilde}{\dot{q}}$ and $\underaccent{\tilde}{p}$. However, on account of \equa{a3}, one may easily prove that $H$ is independent of $\underaccent{\tilde}{\dot{q}}$. Then, if $L = L(\underaccent{\tilde}{q}, \underaccent{\tilde}{\dot{q}})$, one may write
\eqn{
H = H(\underaccent{\tilde}{q}, \underaccent{\tilde}{p}) .
}{a5}
For a conservative mechanical system, the Hamiltonian function coincides with the energy $E$ of the system, namely
\eqn{
H(\underaccent{\tilde}{q}, \underaccent{\tilde}{p}) = T + V = E.
}{a6}
The motion of the mechanical system is described by the \emph{Hamilton equations},
\eqn{
\dot{q}_i = \aderv{H}{p_i}, \qquad \dot{p}_i = - \aderv{H}{q_i} .
}{a7}

\vfill\eject
\renewcommand{\theequation}{B\arabic{equation}}
\section{}\label{appenB}
\textbf{Inversion theorem and the Parseval-Plancherel formula}
\\
Let us define
\eqn{
\phi(p,t) = \frac{1}{\sqrt{2\pi\hbar}} \int_{-\infty}^{\infty} d x\, \psi(x,t) \, e^{-i p x / \hbar} ,
}{b1}
so that we have
\eqn{
\hat{\psi}(x,t) = \frac{1}{\sqrt{2\pi\hbar}} \int_{-\infty}^{\infty} d p\, \phi(p,t) \, e^{i p x / \hbar} = \frac{1}{{2\pi\hbar}} \int_{-\infty}^{\infty} d p\, \int_{-\infty}^{\infty} d x'\, \psi(x',t) \, e^{i p (x - x') / \hbar} \nonumber\\
= \frac{1}{{2\pi\hbar}} \int_{-\infty}^{\infty} d x'\, \psi(x',t) \left[ \int_{-\infty}^{\infty} d p\, e^{i p (x - x') / \hbar} \right] .
}{b2}
Let us now evaluate the integral within square brackets. It is convenient to define 
\eqn{
\int_{-\infty}^{\infty} d p\, e^{i p (x - x') / \hbar} = \int_{-\infty}^{\infty} d p\, e^{i p y / \hbar} = \hbar\, I(y), \quad \text{where} \quad y = x - x' .
}{b3}
With a simple change of integration variable, $k = p/\hbar$, we obtain
\eqn{
I(y) = \frac{1}{\hbar} \int_{-\infty}^{\infty} d p\, e^{i p y / \hbar} = \int_{-\infty}^{\infty} d k\, e^{i k y} = \lim_{a \to 0^+}  \int_{-\infty}^{\infty} d k\, e^{i k y - a k^2} \nonumber\\
=  \lim_{a \to 0^+}  \int_{-\infty}^{\infty} d k\, e^{- a \left( k - \frac{i y}{2 a} \right)^2 - \frac{y^2}{4 a}} =  \lim_{a \to 0^+}  e^{-y^2/(4a)}  \int_{-\infty}^{\infty} d s\, e^{- a s^2}, \quad \text{where} \quad   s =  k - \frac{i y}{2 a} .
}{b4}
Hence, we obtain
\eqn{
I(y) = \lim_{a \to 0^+}  e^{-y^2/(4a)}  \int_{-\infty}^{\infty} d s\, e^{- a s^2} = \lim_{a \to 0^+}  e^{-y^2/(4a)}  \sqrt{\frac{\pi}{a}} = 
\begin{cases}
0 , \quad \text{if} \quad y \ne 0,\\
+\infty , \quad \text{if} \quad y = 0.\\
\end{cases}
}{b5}
Furthermore,
\eqn{
\int_{-\infty}^{\infty} dy \, I(y) = \lim_{a \to 0^+}  \sqrt{\frac{\pi}{a}} \int_{-\infty}^{\infty} dy \, e^{-y^2/(4a)} =  \lim_{a \to 0^+} 2 \sqrt{\frac{\pi}{a}} \sqrt{\pi a} = 2 \pi .
}{b6}
Equations \eqref{b5} and \eqref{b6} imply that
\eqn{
I(y) =  \int_{-\infty}^{\infty} d k\, e^{i k y} = 2 \pi\, \delta(y),
}{b7}
where $\delta(y)$ is Dirac's delta function. By substituting \equasa{b3}{b7} into \equa{b2}, we finally obtain
\eqn{
\hat{\psi}(x,t) = \frac{1}{{2\pi\hbar}} \int_{-\infty}^{\infty} d x'\, \psi(x',t) \, 2\pi\hbar\, \delta(x - x') = \int_{-\infty}^{\infty} d x'\, \psi(x',t) \, \delta(x - x') = \psi(x,t) .
}{b8}
Equations~\eqref{b2} and \eqref{b8} proves the inversion formula employed in \equasa{63}{64}.

In order to prove the Parseval-Plancherel formula, we write
\eqn{
\int_{-\infty}^{\infty} d x\, |\psi(x,t)|^2 = \frac{1}{\sqrt{2\pi\hbar}} \int_{-\infty}^{\infty} dx \left[  \int_{-\infty}^{\infty} dp\; \bar{\phi}(p,t) \, e^{- i p x / \hbar}  \right] \psi(x,t) \nonumber\\
= \int_{-\infty}^{\infty} dp\; \bar{\phi}(p,t) \left[   \frac{1}{\sqrt{2\pi\hbar}} \int_{-\infty}^{\infty} dx \, \psi(x,t)\, e^{- i p x / \hbar}  \right] = \int_{-\infty}^{\infty} dp\; \bar{\phi}(p,t) \, \phi(p,t) \nonumber\\
= \int_{-\infty}^{\infty} dp \, |\phi(p,t)|^2 ,
}{b9}
where \equasa{63}{64} are employed. Equation~\eqref{b9} finally proves the Parseval-Plancherel formula.

\vfill\eject
\section{}
\textbf{Suggested homework}
\\[0.5cm]
\emph{EXERCISE 1}\\
Consider a quantum particle in the following one--dimensional potential,
\[
V(x) = 
\begin{cases}
\infty, \quad &x < 0, \\
\lambda x, &x > 0 ,
\end{cases}
\]
where $\lambda$ is a positive constant.\\
(a) Determine the energy spectrum. Is it continuous or discrete?\\
(b) Plot the wave functions and the probability densities of the first 3 stationary states.\\
(c) Could $V(x)$ be a good representation of some physically interesting system?\\
\vspace{1cm}
~\\
\emph{EXERCISE 2}\\
Consider a quantum particle in the following one--dimensional potential,
\[
V(x) = 
\begin{cases}
\infty, \quad &x < 0, \\
0, &0 < x < a ,\\
V_0, &a < x < b ,\\
0, & x > b,
\end{cases}
\]
where $V_0$ is a positive constant, and $b - a > 0$ is the thickness of the potential barrier.\\
(a) Determine the energy spectrum. Is it continuous or discrete?\\
(b) Determine the transmission coefficient for $E < V_0$.\\
\vspace{1cm}
~\\
\emph{EXERCISE 3}\\
Consider a free quantum particle, i.e. a particle with Hamiltonian,
\[
\hat{H} = \frac{\hat{p}^2}{2 m} .
\]
Evaluate the commutator
\[
[ \hat{x}, \hat{H} ]
\]
and determine the lower bound on the product of the two uncertainties $\Delta x$ and $\Delta E$, where $\Delta x$ is the uncertainty associated with $\hat{x}$ and $\Delta E$ is the uncertainty associated with $\hat{H}$.\\
Comment on the result.

\vfill\eject


\begin{thebibliography}{7}
\providecommand{\natexlab}[1]{#1}
\providecommand{\url}[1]{\texttt{#1}}
\expandafter\ifx\csname urlstyle\endcsname\relax
  \providecommand{\doi}[1]{doi: #1}\else
  \providecommand{\doi}{doi: \begingroup \urlstyle{rm}\Url}\fi

\bibitem[Gasiorowicz(1974)]{Gasiorowicz1974}
S.~Gasiorowicz.
\newblock \emph{Quantum Physics}.
\newblock Wiley, 1974.

\bibitem[Phillips(2003)]{Phillips2003}
A.~C. Phillips.
\newblock \emph{Introduction to Quantum Mechanics}.
\newblock Wiley, 2003.

\bibitem[Dirac(1981)]{Dirac1981}
P.~A.~M. Dirac.
\newblock \emph{The Principles of Quantum Mechanics}.
\newblock Clarendon Press, 1981.

\bibitem[Baggott(2011)]{Baggott2011}
J.~Baggott.
\newblock \emph{The Quantum Story: A History in 40 Moments}.
\newblock Oxford University Press, 2011.

\bibitem[Messiah(1999)]{Messiah1999}
A.~Messiah.
\newblock \emph{Quantum Mechanics}.
\newblock Dover, 1999.

\bibitem[Peres(1995)]{Peres1995}
A.~Peres.
\newblock \emph{Quantum Theory: Concepts and Methods}.
\newblock Kluwer Academic Publishers, 1995.

\bibitem[Sakurai(1994)]{Sakurai1994}
J.~J. Sakurai.
\newblock \emph{Modern Quantum Mechanics}.
\newblock Addison-Wesley, 1994.

\bibitem[Susskind~and~Friedman(2015)]{Susskind2015}
L.~Susskind and A.~Friedman.
\newblock \emph{Quantum Mechanics:~The Theoretical Minimum}.
\newblock Penguin, 2015.



\end{thebibliography}
\end{document}